\documentclass[hidelinks,runningheads]{llncs}

\usepackage[english]{babel}
\usepackage[T1]{fontenc}
\usepackage[utf8]{luainputenc}

\usepackage{algorithmicx}
\usepackage{algpseudocode}

\usepackage{amsmath}
\usepackage{amssymb}
\usepackage{thm-restate}

\usepackage[english]{varioref}
\usepackage[final]{hyperref}

\usepackage{ifthen}

\newcommand*{\HASAPPENDIX}{}
\newcommand{\withappendix}[1]{\ifthenelse{\isundefined{\HASAPPENDIX}}{}{#1}}
\newcommand{\withauthors}[1]{\ifthenelse{\isundefined{\HASAPPENDIX}}{#1}{}}

\usepackage{xargs}
\usepackage{todonotes}
\newcommandx{\mytodo}[2][1=]{
\todo[linecolor=red,backgroundcolor=red!25,bordercolor=red,#1]{#2}{}}

\usepackage{tikz}
\usetikzlibrary{automata,positioning,arrows,backgrounds,calc,decorations.pathreplacing,decorations.pathmorphing,shapes,fit}
\tikzset{every state/.style={minimum size=0pt}}

\newcommand{\cupdot}{\mathbin{\mathaccent\cdot\cup}}
\newcommand{\Nat}{\mathbb{N}}

\newcommand{\op}[1]{\ensuremath{\mathsf{#1}}}
\newcommand{\mc}[1]{\mathcal{#1}}
\newcommand{\parent}{\op{\uparrow}}

\newcommand{\scc}{\op{scc}}

\usepackage{graphicx}

\begin{document}
\title{New Optimizations and Heuristics for Determinization of Büchi Automata
\thanks{The final authenticated publication is
available online at \url{https://doi.org/10.1007/978-3-030-31784-3_18}}}
\author{Christof Löding \and
Anton Pirogov
\thanks{This work is supported by the German research council (DFG) Research Training Group 2236 UnRAVeL}}
\institute{RWTH Aachen University, Templergraben 55, 52062 Aachen, Germany \\
\email{\{loeding,pirogov\}@cs.rwth-aachen.de} }
\maketitle
\begin{abstract}
  In this work, we present multiple new optimizations and heuristics for the
  determinization of Büchi automata that exploit a number of semantic and structural
  properties, most of which may be applied together with any determinization procedure.
  We built a prototype implementation where all the presented heuristics can be freely
  combined and evaluated them, comparing our implementation with the state-of-the-art tool
  \texttt{spot} on multiple data sets with different characteristics. Our results show
  that the proposed optimizations and heuristics can in some cases significantly decrease
  the size of the resulting deterministic automaton.

\end{abstract}

\section{Introduction}
\label{fig:introduction}

Nondeterministic Büchi automata (NBA) \cite{buchi1966symposium} are a
well-established formalism for the representation of properties of
non-terminating executions of finite state programs, and arise often
as a low-level representation obtained by translation of some formula
describing the desired behaviour in a logic like e.g. LTL
\cite{pnueli1977temporal}. In the field of model checking, Büchi
automata are an important tool in verification algorithms for finite
state systems (see, e.g., \cite{baier2008principles}). In order to
capture the full class of regular $\omega$-languages, Büchi automata
need to be nondeterministic. In some applications, however, the
algorithms require the property to be represented by a deterministic
automaton. For example, in probabilistic model checking, the natural
product of an automaton with a Markov chain requires the automaton to
be deterministic in order to produce again a Markov chain. Therefore,
many algorithms in this setting need deterministic automata as
input (see, e.g., \cite[Section~10.3]{baier2008principles}). Another
example is the problem of synthesis of finite state systems from
$\omega$-regular specifications, where the specification can be
translated into a game over a finite graph based on a deterministic
parity automaton \cite{Thomas08,MeyerSL18}
.

For this reason, the determinization of Büchi automata is a
well-researched problem. In order to obtain a deterministic automaton
model that is as expressive as nondeterministic Büchi automata, one
needs to use a more expressive acceptance condition such as Muller,
Rabin, Streett, or parity conditions (see, e.g.,
\cite{Thomas90,VardiW07}).  The first determinization construction by
McNaughton \cite{mcnaughton1966testing} resulted in doubly-exponential
Muller automata, whereas the first asymptotically optimal construction
was presented by Safra \cite{safra1988complexity}, which yields Rabin
automata with $2^{\mc{O}(n\log n)}$ states. Since then, modifications
of Safra's construction have been proposed in order to improve the
constants in the exponent of the state complexity
\cite{piterman2006nondeterministic,schewe2009tighter}, resulting in a
construction for which tight lower bounds exist
\cite{ColcombetZ09,schewe2009tighter}. In particular, Piterman was the
first to present a direct construction of a deterministic parity
automaton (DPA) \cite{piterman2006nondeterministic}, which is described in a
slightly different way  in \cite{schewe2009tighter}. Another
variant of Safra's construction is presented in
\cite{redziejowski2012improved}.

While the Safra construction (and its variants) is the most famous
determinization construction for Büchi automata, there is another
approach which can be derived from a translation of alternating tree
automata into nondeterministic ones by Muller and Schupp
\cite{muller1995simulating}. Determinization constructions for Büchi
automata based on the ideas of Muller and Schupp have been described
in \cite{kahler2008complementation,fogarty2015profile,FismanL15}. The
two types of constructions, Safra and Muller/Schupp, are unified in
\cite{unidet}.

While in theory, constructions with tight upper and lower bounds have
been achieved, there is a lot of room for optimizations when
implementing determinization constructions. A first implementation of
Safra's construction in the version of
\cite{piterman2006nondeterministic} is \texttt{ltl2dstar}, presented in
\cite{klein2005diploma}. While \texttt{ltl2dstar} already implements
some optimizations and heuristics in order to reduce the size of the
output automaton, the resulting automata are still quite large, even
for small input Büchi automata.
The current state-of-the-art tool for determinization of Büchi
automata is part of the library \texttt{spot} \cite{duret2016spot},
which implements the variant of Safra's construction presented in
\cite{redziejowski2012improved}. Furthermore, \texttt{spot} also
implements more optimizations to reduce the size of the output
automata.

The contribution of this paper is the identification of new heuristics
for reducing the size of DPA produced
by determinization from Büchi automata and a simple framework for their implementation.
More specifically, our contributions can roughly be described as follows:
\begin{itemize}
\item We present a modularized version of the construction in \cite{unidet} which enables
  the integration of SCC-based heuristics.
\item By exploiting properties of the used construction,
  we make stronger use of language inclusions between states of
    the given NBA (e.g. obtained by simulation), permitting to use inclusions between states in the same SCC.
\item We treat specific types of SCCs in the NBA in a special way,
  namely those without accepting states, those with only accepting
  states, and those with a deterministic transition relation.
\item The construction from \cite{unidet} leaves some freedom in the
  choice of successor states. Our implementation admits different
  options for this successor choice, in particular those leading to
  the constructions of Safra and of Muller/Schupp.

  Furthermore, we have an optimization that exploits this freedom of
  successor selection by checking whether admissible successor states
  have already been constructed before adding a new state.
\item We make use of language equivalences of states in the
  constructed DPA in order to remove some SCCs of the resulting
  DPA. These language equivalences are derived during construction
  time from the powerset automaton for the given NBA.
\item We propose to use known minimization techniques as post
  processing, which to the best of our knowledge, have not yet been
  used in this context. We first minimize the number of priorities of
  the DPA using an algorithm from \cite{carton1999computing}, and then
  reduce the size of the resulting DPA by Hopcroft's algorithm
  \cite{hopcroft1971n}, treating it as a finite automaton that outputs
  the priorities.
\item
  We have evaluated the combination
  of different heuristics on different types of data-sets (randomly
  generated NBAs, NBAs constructed from random LTL formulas, NBAs
  generated from some families of LTL formulas taken from the
  literature, and NBAs obtained from specifications of the competition
  SYNTCOMP \cite{jacobs20174th}), and compared the size of the resulting automata with the
  ones produced by \texttt{spot}.

\end{itemize}

This work is organized as follows. After some preliminaries in
Section~\ref{sec:preliminaries}, we sketch the general construction on which we based our
implementation in Section~\ref{sec:construction}, and then describe our optimizations and
heuristics in Section~\ref{sec:heuristics}. In Section~\ref{sec:experiments} we present
our experiments and in Section~\ref{sec:conclusion} we conclude.
 \vspace{-3mm}
\section{Preliminaries}
\label{sec:preliminaries}
\vspace{-3mm}
First we briefly review basic definitions concerning $\omega$-automata and $\omega$-languages.
If $\Sigma$ is a finite alphabet, then $\Sigma^\omega$ is the set of all infinite
\emph{words} $w=w_0w_1\ldots$ with $w_i\in\Sigma$. For $w\in \Sigma^\omega$ we denote by
$w(i)$ the $i$-th symbol $w_i$. For convenience, we write $[n]$ for the set of natural
numbers $\{1,\ldots,n\}$.
A \emph{Büchi automaton} $\mc{A}$ is a tuple $(Q, \Sigma, \Delta, Q_0,
F)$, where $Q$ is a finite set of states, $\Sigma$ a finite alphabet, $\Delta \subseteq
Q\times \Sigma \times Q$ is the transition relation and $Q_0, F \subseteq Q$ are the sets
of initial and final states, respectively. When $Q$ is understood and $X\subseteq Q$, then
$\overline{X} := Q \setminus X$.
We write $\Delta(p,a) := \{q \in Q \mid (p,a,q) \in \Delta \}$ to denote the set of
\emph{successors} of $p \in Q$ on symbol $a \in \Sigma$, and $\Delta(P,a)$ for $\bigcup_{p\in
P}\Delta(p,a)$.
A \emph{run} of an automaton on a word $w\in \Sigma^\omega$ is an infinite sequence of
states $q_0, q_1, \ldots$ starting in some $q_0 \in Q_0$ such that $(q_i, w(i), q_{i+1})
\in \Delta$ for all $i\geq 0$.
As usual, an automaton is \emph{deterministic} if $|Q_0|=1$ and $|\Delta(p,a)|\leq 1$ for
all $p\in Q, a\in\Sigma$, and \emph{non-deterministic} otherwise. In this work, we assume
Büchi automata to be non-deterministic and refer to them as NBA.
A \emph{(transition-based) deterministic parity automaton} (DPA) is a deterministic
automaton $(Q,\Sigma,\Delta,Q_0,c)$ where instead of $F\subseteq Q$ there is a
\emph{priority function} $c : \Delta \to \mathbb{N}$ assigning a natural number to each
transition.

A run of an NBA is \emph{accepting} if it contains infinitely many accepting states.
A run of a DPA is accepting if the smallest priority of transitions along the
run which appears infinitely often is even.
An automaton $\mc{A}$ \emph{accepts} $w\in \Sigma^\omega$ if there
exists an accepting run on $w$, and the language $L(\mc{A}) \subseteq
\Sigma^\omega$ \emph{recognized} by $\mc{A}$ is the set of all
accepted words. If $P$ is a set of states of an automaton, we write
$L(P)$ for the language accepted by this automaton with initial state
set $P$. For sets consisting of one state $q$, we write $L(q)$ instead
of $L(\{q\})$.  We sometimes refer to states of a DPA that is obtained
by a determinization construction as \emph{macrostates} to distinguish
them from the states of the underlying Büchi automaton.

We write $p \overset{x}{\to} q$ if there exists a path from $p$ to $q$ labelled
by $x\in \Sigma^{+}$ and $p \to q$ if there exists some $x$ such that $p\overset{x}{\to} q$.
The \emph{strongly connected component (SCC)} of $p\in Q$ is $\scc(p) := \{q\in Q \mid p=q
\text{\ or\ } p \to q \text{\ and\ } q \to p \}$.
The set $\op{SCCs}(\mc{A}) := \{\scc(q) \mid q\in Q\}$ is the set of all SCCs and partitions $Q$.
A component $C\in\op{SCCs}(\mc{A})$ is \emph{trivial} if $C=\{q\}$ for
some $q\in Q$ and $q \not\to q$.  $C$ is \emph{bottom} if $p \to q$
implies $q \in C$ for all $p\in C$ and $q\in Q$. In a Büchi automaton,
$C$ is \emph{rejecting} if it is trivial or contains no accepting
states, and \emph{accepting} if it is not trivial and all cycles
inside $C$ contain an accepting state.
If an SCC is neither accepting or rejecting, it is \emph{mixed}.
Notice that rejecting and accepting components are often called \emph{inherently weak} in the
literature (e.g. \cite{boigelot2001use}).
Finally, $C$ is \emph{deterministic} if $|\Delta(p,a)\cap C|\leq 1$ for all $p\in C$ and $a \in \Sigma$.
An NBA is \emph{limit-deterministic}, if all its non-rejecting SCCs are deterministic and
cannot reach non-deterministic SCCs again.
 \vspace{-3mm}
\section{Construction}
\label{sec:construction}

\vspace{-2mm}
For our prototype implementation, we applied new optimizations and
heuristics to an adaptation of a recent generalized determinization
construction from NBA to DPA that was presented in \cite{unidet}.
This construction unifies the constructions of Safra and of Muller and
Schupp, and also introduces new degrees of freedom, which we
exploit in one of our heuristics (see Section~\ref{subsec:NBASCCs}).

Let $\mc{A} = (Q, \Sigma, \Delta, Q_0,
F)$ be the NBA to be determinized.
The macrostates $(\alpha, t)$ in the deterministic automaton (called \emph{ranked slices})
are tuples of disjoint non-empty sets $t := (S_1, S_2, \ldots, S_n)$ equipped with a
bijection $\alpha : [n] \to [n]$ that assigns to each set $S_i$ the
rank $\alpha(i)$. These ranks are used to define the priorities of the
transitions.
When reading symbol $a \in \Sigma$ in macrostate $(\alpha, t)$, the
successor $(\alpha', t')$ is obtained by applying the
successive operations $\op{step}$, $\op{prune}$, $\op{merge}$ and $\op{normalize}$.
An overview of the complete transition is sketched in
Figure~\ref{fig:constructionsketch}
and we refer to \cite{unidet} for more details.

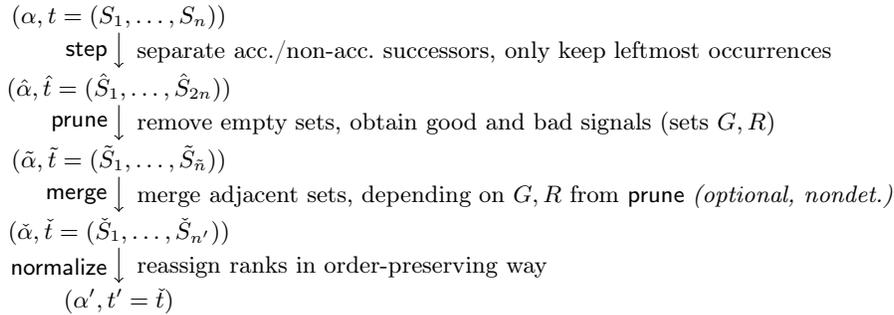
\begin{figure}[h]
  \begin{center}
    \begin{tikzpicture}[baseline={([yshift=-.5ex]current bounding box.center)},
      shorten >=1pt,node distance=5mm,inner sep=1pt,auto]
      \node              (t1) {$(\alpha, t=(S_1,\ldots, S_n))$};
      \node[below=of t1] (t2) {$(\hat{\alpha}, \hat{t}=(\hat{S}_1,\ldots, \hat{S}_{2n}))$};
      \node[below=of t2] (t3) {$(\tilde{\alpha}, \tilde{t}=(\tilde{S}_1,\ldots, \tilde{S}_{\tilde{n}}))$};
      \node[below=of t3] (t4) {$(\check{\alpha}, \check{t}=(\check{S}_1,\ldots, \check{S}_{n'}))$};
      \node[below=of t4] (t5) {$(\alpha', t'=\check{t})$};

      \path[->]
        (t1) edge node[left]{$\op{step}~$} node[xshift=2mm] {separate
        acc./non-acc. successors, only keep leftmost occurrences} (t2)
        (t2) edge node[left]{$\op{prune}~$} node[xshift=2mm]{remove empty sets, obtain good and bad signals (sets $G, R$)} (t3)
        (t3) edge node[left]{$\op{merge}~$} node[xshift=2mm] {merge adjacent sets,
        depending on $G,R$ from $\op{prune}$ \emph{(optional, nondet.)}} (t4)
        (t4) edge node[left]{$\op{normalize}~$} node[xshift=2mm] {reassign ranks in order-preserving way} (t5)
          ;
    \end{tikzpicture}
    \vspace{-7mm}
  \end{center}
  \caption{Sketch of some transition $(\alpha,t)\overset{a\in\Sigma}{\rightarrow}(\alpha',t')$
  of the determinization construction.
  The nondeterminism in the optional $\op{merge}$ operation enables the simulation of
  different constructions and provides freedom for optimizations.}
  \label{fig:constructionsketch}
\end{figure}

\vspace{-3mm}
\subsection{A modular variant of the construction}
\label{sec:modularsketch}

\vspace{-1mm}
We present a generalization of the construction above which
offers a clean framework to implement our heuristics, specifically the heuristics that
exploit the SCC structure of the Büchi automaton. This variant of the construction
runs multiple interacting instances of the construction above in parallel, where each
instance is essentially a ranked slice as above, but manages only an
SCC of the Büchi automaton (or a union of SCCs).
A macrostate in this modular construction consists of not one, but (in general) multiple
tuples with a global ranking function (see illustration in
Figure~\ref{fig:modular}).
Furthermore,
for some heuristics it is not needed to have a rank assigned to certain
states of the Büchi automaton.  We therefore additionally store a
separate ``buffer set'' in each macrostate with no assigned rank where such Büchi states
can be placed in. The states in the buffer set are not used to track an accepting run
directly, but they may reach successors later that should be managed inside of a tuple.

\vspace{-6mm}
\begin{figure}[h]
  \[
    (\{q_1,q_3\}^5,\{q_2\}^3) \mid (\{q_4\}^1, \{q_6\}^6, \{q_5\}^4) \mid (\{q_7,q_8,q_9\}^2) \mid\mid \{q_{10}\}
  \]
\vspace{-6mm}
  \caption{Illustration of some macrostate of the modularized determinization procedure.
  In this example, assume that the underlying automaton has components with the states
  $q_1,q_2,q_3$ and $q_4,q_5,q_6$ and $q_7,\ldots,q_{10}$, respectively. The superscripts denote the global
  rank of the corresponding set. State $q_{10}$ is currently not considered to track
  accepting runs, so it is stored in the separated buffer set and has no assigned rank.
  The vertical bars separate the different components from each other.
  }
  \label{fig:modular}
\end{figure}

\vspace{-5mm}
To ensure the correct interaction of the tuples,
the ranks assigned to sets must be unique in the whole macrostate.
Apart from this difference, the
construction above is applied to each tuple individually. Whenever a
transition in the Büchi automaton moves from an SCC $C_1$ to a
different SCC $C_2$,
and the same state is not also reached by some transition inside of $C_2$,
then the resulting
state is added in a new right-most set in the tuple responsible for
SCC $C_2$. The rank of this new set is the least important (i.e., largest in the tuple) one.

The correctness of this modular approach essentially follows from the fact that each accepting run
eventually stays in the same SCC of the Büchi automaton forever, i.e., the states along
this run eventually stay in the same tuple forever and are managed by the construction above
without interference.
\withappendix{
More details about the required adaptations can be found in Appendix~\ref{app:modular}.
}
 \vspace{-4mm}
\section{Optimizations and Heuristics}
\label{sec:heuristics}

\vspace{-3mm}
In the following, we describe the optimizations and heuristics we suggest to consider
during determinization of Büchi automata.

\vspace{-3mm}
\subsection{Using known language inclusions of Büchi states}
\label{subsec:langinc}

\vspace{-2mm}

We are aware of only two rather simple optimizations that are exploiting known language inclusion
relations between states of the Büchi automaton---a ``true-loop'' optimization from
\texttt{ltl2dstar} \cite{klein2005diploma} that syntactically identifies NBA states that can be treated as
accepting sinks, and an optimization used in \texttt{spot}
\cite{duret2016spot} that is restricted to pairs of states from different SCCs of the NBA, which we refer to as ``external'' language inclusions.
In general, it is nontrivial to use the language inclusions between pairs of states in the same SCC of the NBA, because
both of them can occur infinitely often in the same run.

However, there is a possibility to use the language inclusion relation more generally
within the same SCC in a safe manner by exploiting properties of the determinization
construction from \cite{unidet}. The result is what we call the ``internal'' language
inclusion optimization, which works for the original non-modular version of the
construction, which uses a single ranked tuple $(\alpha,t)$ per macrostate, as follows.
Let $Q_t := \bigcup_{i=1}^{n} S_{i}$ for $t=(S_1,\ldots,S_n)$ and let
the function $\op{idx}_t : Q_{t} \to [n]$ map each state $q\in Q_{t}$ to the
tuple index $i$ such that $q \in S_{i}$.

\vspace{-1mm}
\begin{restatable}{theorem}{thmlanginclusions}
  \label{thm:inclusions}
  Let $p,q \in Q$ with $L(p) \subseteq L(q)$
  and let $(\alpha, t)$ be a macrostate of the determinization construction
  such that $p,q \in Q_t$.
  If $q$ is to the left of $p$ (i.e. $\op{idx}_t(q) < \op{idx}_t(p)$),
  then $p$ may be omitted from the macrostate without changing the language of the determinized automaton.
\end{restatable}

\withappendix{
A sketch of the correctness proof and required modifications for the modularized
construction can be found in Appendix~\ref{app:inclusions}.
}

Since the internal or external language inclusion optimization change
the structure of the macrostates, the overall structure of the
constructed DPA might change significantly.
For this reason, there is no guarantee that this technique cannot increase the number
of states in some cases.
In our experiments we found no instance where this is the case, i.e., in practice they only
decrease the number of states.

\vspace{-3mm}
\subsection{Using properties of SCCs in the NBA}
\label{subsec:NBASCCs}

In the following, we describe some heuristics to simplify the
treatment of SCCs in the NBA that are rejecting, accepting, or
deterministic.
Earlier, we described how the determinization construction can be performed in a modular
way. The following heuristics can be implemented very cleanly in that framework, but in
principle could also be used with other constructions.

\vspace{-3mm}
\paragraph{Rejecting SCCs:}

It is known that one can keep states from rejecting NBA SCCs separate from the
determinization construction, as no accepting run can visit them infinitely often.
This can be realized in our modular construction by keeping states of
rejecting SCCs outside of the determinization tuples in the buffer set.
A related optimization is already implemented in \texttt{spot} as a modification of the
construction from \cite{redziejowski2012improved}.

\vspace{-2mm}
\paragraph{Accepting SCCs:}
For accepting SCCs, it is sufficient to check that at least one run eventually stays in the
same SCC forever. For this, an adaptation of the Miyano-Hayashi
construction \cite{miyano1984alternating} (often called ``breakpoint''-construction) can
be used, which requires to manage only two different sets---a track
set and a background set. In a transition, the track set is updated to
all successors of the current track set that are in an accepting
SCC. The background set contains all other states from accepting
SCCs that are reached in this transition.
This pair of sets has one rank assigned by the global ranking
function. As long as the track set is non-empty, the rank signals a
good event. If the track set becomes empty in a transition, then the rank
signals a bad event, the background set becomes the new
track set, and a new rank is assigned that is larger than all the
ranks that survived the last transition.
It is not very difficult to see that this construction correctly
detects runs of the NBA that remain inside an accepting SCC.

This heuristic is realized in our modular construction by delaying the movement of states into the
corresponding component tuple until a ``breakpoint'' happens, thereby ensuring that the
tuple responsible for all accepting components always contains at most one non-empty set.
If the two heuristics for rejecting and accepting SCCs are used, and
the input NBA is a weak automaton (in which all SCCs are
either rejecting or accepting), then one obtains the
pure breakpoint construction, which is used, e.g., in
\cite{boigelot2001use} to determinize weak automata. The overall state
complexity is then reduced from $2^{\mc{O}(n\log n)}$ to $3^{n}$.
\withappendix{
More details can be found in Appendix~\ref{app:accsccs}.
}

\vspace{-2mm}
\paragraph{Deterministic SCCs:}

If an SCC has both accepting and rejecting states, but is deterministic, a run never
branches into multiple runs as long as the successors stay in the same component. Hence, the
number of different runs can only decrease or stay the same. This excludes the possibility
that an accepting state is visited infinitely often by different runs, but not by a single
infinite run. Therefore, whenever a set of states in a deterministic SCC is reached from
some other SCC, it suffices to add the states to the tuple which is responsible for this
SCC with a new rank, but in the following steps there is no need to refine this set, i.e.
separate accepting from non-accepting states, as this is only required for
filtering out infinite non-accepting runs. For a good event to be signalled by such a
component in the construction, it suffices that a set just contains an accepting state. If
this set eventually never becomes empty and infinitely often contains accepting states,
clearly at least one of the finitely many runs evolving in this set must visit accepting
states infinitely often. Applying this heuristic on limit-deterministic NBA simplifies
the determinization procedure to a variant of the construction described in
\cite{esparza2017ltl}, because then the importance ordering given by ranks coincides with
the tuple index ordering, effectively mimicking the tuples as used in \cite{esparza2017ltl}.
\withappendix{
More details can be found in Appendix~\ref{app:detsccs}.
}

It should be pointed out that all these SCC-based heuristics can in fact increase the
number of states for specific instances,
even though they restrict the state space in general
\withappendix{
  (e.g. see Appendix~\ref{app:experiments}, Figure~\ref{fig:badforheu})
}
and therefore must be applied with greater care.

\vspace{-3mm}
\subsection{Smart successor selection}
\label{subsec:smart}

The determinization construction in \cite{unidet} permits, in general, multiple
valid successor macrostates because of the freedom in the $\op{merge}$ operation.
Determinization in practice usually works by starting in some initial state and fully exploring the
state space which is reachable using the transitions prescribed by some construction.
The natural optimization to be derived from the non-determinism of
$\op{merge}$ is to check for each transition whether a permissible
successor state has already been constructed. In this case, just a new edge is added, and
a new state is constructed only if no viable existing successor could be found.

The question is then how to find such a state as efficiently as
possible. A naive approach would be to check for every already visited
state whether it satisfies the constraints on the shape of a successor
for the current transition. This would incur a blow-up in the running
time that is quadratic in the size of the output automaton, which can
be very large in general.  A better approach is to structure the set
of already visited states in a reasonable way that accelerates this
search. We achieve this by managing a trie where each node corresponds
to a macrostate. Each trie node is labelled by a set $S \subseteq Q$
of states of the NBA and can be marked to denote whether the
corresponding macrostate already exists in the DPA.  There exists a
simple bijection that maps each ranked slice into a sequence of sets
(without ranks) which uniquely determines a node in the trie. The
resulting sequence of sets $S_1,S_2,\ldots,S_n$ is the ``word'' over
the alphabet $2^Q$ which is inserted into the trie if the
corresponding macrostate is added to the DPA (the sequence
$S_1,S_2,\ldots,S_n$ is different from the tuple in the ranked slice).

When constructing a successor state, we first apply the operations
$\op{step}$ and $\op{prune}$. From that, we obtain the minimal active
rank $k$.  Under the constraints of the construction sketched in
Section~\ref{sec:construction}, the $\op{merge}$ operation can now
merge sets with rank at least $k$  which in turn means that the sets
with rank smaller than $k$ do not change during $\op{merge}$. This sequence
of sets of rank smaller than $k$ determines a trie node such that all
possible successors that can be constructed by $\op{merge}$ are below this node.

Only a simple check must
be performed on remaining candidate macrostates that are found in the trie, and all
required steps can be carried out efficiently using bit-set operations,
so that the overall slowdown incurred by this optimized successor selection is very
moderate.
\withappendix{
More details about this optimization can be found in Appendix~\ref{app:smartsucc}.
}

This optimization only prevents new redundant states from being
constructed, thus the resulting automaton can never become larger than
without this optimization enabled. However, in combination with other
heuristics, like the post-processing described in
Section~\ref{subsec:minimize}, it might also have a negative effect.

\vspace{-4mm}
\subsection{The benefits of the powerset structure}
\label{subsec:topological}

\vspace{-2mm}
It is a well-known fact that language equivalent states of a
deterministic or non-deterministic $\omega$-automaton cannot be
merged, in general, without changing the accepted language. However,
if the two states are in different SCCs, it is possible to remove one
of them. The aim of the heuristic that we describe next, is to exploit
this fact during construction of the DPA. Formally, we rely on the
following proposition, which we consider folklore.
\begin{proposition}
  \label{prop:keepdeep}
  Let $s,s'$ be states of a DPA. If $L(s) = L(s')$ and $s
  \not\rightarrow s'$, then all incoming edges of $s'$ can be redirected
  to $s$ without changing the recognized language (and $s'$ can be removed, since
  it becomes unreachable).
\end{proposition}
This implies that the whole SCC containing such a state $s'$ can be
removed because the language equivalence holds for all
the successor states of $s$ and $s'$ as well, due to the deterministic
transitions.

In order to use this insight in the determinization construction, we
need to detect language equivalences of states of the DPA at
construction time. We do this by using the \emph{powerset structure}
$\op{PS}(\mc{A})$ of the given NBA $\mc{A}$,
which is the transition system with nodes from $2^Q$ obtained by using the NFA powerset construction
on the Büchi automaton, with $Q_0$ as initial state and a deterministic transition
function $\delta^{\op{PS}(\mc{A})}(P,a) := \Delta(P,a)$. Our
optimization is based on the following simple observation.
\begin{proposition}\label{prop:powerset-language}
Let $\mc{B}$ be a DPA that is equivalent to the NBA $\mc{A}$. Let $u$
be a finite word such that in $\mc{B}$ the state $s$ is reached via
$u$, and in $\op{PS}(\mc{A})$ the set $P$ is reached via $u$. Then
  $\mc{B}$ accepts from $s$ the same language as $\mc{A}$ accepts from $P$.
\end{proposition}
This implies that states of the DPA that can be reached simultaneously
with the same set $P$ in $\op{PS}(\mc{A})$ are language equivalent.
In combination with Proposition~\ref{prop:keepdeep}, we obtain that a
DPA equivalent to $\mc{A}$ only needs one SCC per SCC in
$\op{PS}(\mc{A})$.
\begin{corollary}
  \label{cor:topo}
Let $\mc{A}$ be an NBA and $\op{PS}(\mc{A})$ its powerset structure. Then there exists
a DPA $\mc{B}$ recognizing the same language such that for each SCC of $\op{PS}(\mc{A})$
there is at most one SCC in $\mc{B}$.
\end{corollary}

Based on this observation, we separately construct for each SCC $C$ of $\op{PS}(\mc{A})$ one SCC of the DPA, as explained in the following.
The construction picks some
node $S \in C$ (which is a set of NBA states), and starts the determinization
procedure using $S$ as initial states. Furthermore, it tracks for each macrostate $s$ also the
corresponding node $P_s$ in $C$. On each transition, the successor macrostate for the DPA is only constructed, if the corresponding transition in $\op{PS}(\mc{A})$ stays in $C$.

This construction gives us a  DPA $\mc{B}_C$ that can be partial in the case of non-bottom SCCs $C$ of $\op{PS}(\mc{A})$, with
``holes'' for transitions that exit $C$.
Note that $\mc{B}_C$ might consist of several SCCs itself. Based on Propositions~\ref{prop:keepdeep} and~\ref{prop:powerset-language}, we only need to keep one bottom SCC of $\mc{B}_C$.

In order to complete the missing transitions, consider such a transition leading outside of $C$ to some $P'$ in an SCC $C'$ of $\op{PS}(\mc{A})$. We assume that we have already done the determinization for $C'$ (starting at the bottom SCCs of $\op{PS}(\mc{A})$ and then going backwards), so there already exists a macrostate that corresponds to $P'$, and we can let the transition of the DPA point to that macrostate. This idea is illustrated in Figure~\ref{fig:toposketch}.

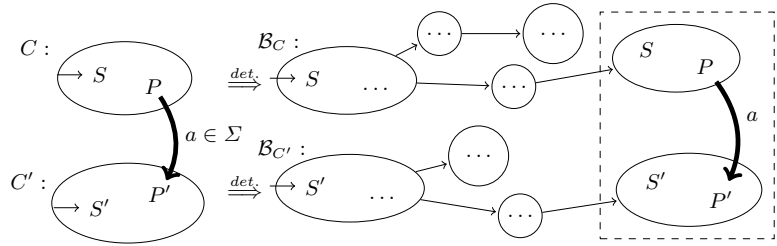
\begin{figure}[h]
  \centering
  \scalebox{0.9}{
    \begin{tikzpicture}[baseline={([yshift=-.5ex]current bounding box.center)},
      shorten >=1pt,node distance=1mm and 4mm,auto]
      \node[initial,initial text=] (s) {$S$};
      \node[below right=of s,yshift=3mm,inner sep=1pt] (t) {$P$};
      \node[shape=ellipse,draw,minimum size=1cm,fit={(s) (t)}] (ps1) {};
      \node[left=of s,yshift=4mm] (lab1) {$C:$};

      \node[initial,initial text=,below=1.5cm of s] (s2) {$S'$};
      \node[above right=of s2,yshift=-3mm,inner sep=1pt] (t2) {$P'$};
      \node[shape=ellipse,draw,minimum size=1cm,fit={(s2) (t2)}] (ps2) {};
      \node[left=of s2,yshift=4mm] (lab2) {$C':$};

      \node[right=of ps1] (arr1) {$\overset{det.}{\Longrightarrow}$};
      \node[above right=of arr1,yshift=-1mm,xshift=-7mm] (lab3) {$\mc{B}_C:$};
      \node[below=1cm of arr1] (arr2) {$\overset{det.}{\Longrightarrow}$};
      \node[above right=of arr2,yshift=-1mm,xshift=-7mm] (lab4) {$\mc{B}_{C'}:$};

      \node[initial,initial text=,right=of arr1] (ds1) {$S$};
      \node[below right=of ds1,yshift=3mm] (tmp1) {$\ldots$};
      \node[shape=ellipse,draw,minimum size=1cm,inner sep=0pt,fit={(ds1) (tmp1)}] (detps1) {};
      \node[state,below right=-0.5cm and 1.5cm of detps1] (detps2) {$\ldots$};
      \node[state,above right=of detps1] (detps4) {$\ldots$};
      \node[state,right=of detps4,inner sep=2mm,xshift=5mm] (detps5) {$\ldots$};
      \node[below right=-1cm and 1.5cm of detps2] (ds2) {$S$};
      \node[below right=of ds2,yshift=3mm] (dt1) {$P$};
      \node[shape=ellipse,draw,minimum size=1cm,inner sep=0pt,fit={(ds2) (dt1)}] (detps3) {};

      \node[initial,initial text=,right=of arr2] (ds21) {$S'$};
      \node[below right=of ds21,yshift=3mm] (tmp21) {$\ldots$};
      \node[shape=ellipse,draw,minimum size=1cm,inner sep=0pt,fit={(ds21) (tmp21)}] (detps21) {};
      \node[state,above right=-1cm and 1.5cm of detps21] (detps22) {$\ldots$};
      \node[state,above right=of detps21,yshift=-3mm,xshift=4mm,inner sep=2mm] (detps24) {$\ldots$};
      \node[below right=-1cm and 1.5cm of detps22] (ds22) {$S'$};
      \node[below right=of ds22,yshift=3mm] (dt2) {$P'$};
      \node[shape=ellipse,draw,minimum size=1cm,inner sep=0pt,fit={(ds22) (dt2)}] (detps23) {};

      \node[shape=rectangle,draw,dashed,minimum size=2cm,inner sep=5pt,fit={(detps3) (detps23)}] (keep) {};

      \draw[->] (t) edge[bend left,line width=2pt] node[yshift=-2mm] {$a\in\Sigma$} (t2);

      \draw[->] (detps1) edge (detps2);
      \draw[->] (detps2) edge (detps3);
      \draw[->] (detps1) edge (detps4);
      \draw[->] (detps4) edge (detps5);

      \draw[->] (detps21) edge (detps22);
      \draw[->] (detps22) edge (detps23);
      \draw[->] (detps21) edge (detps24);

      \draw[->] (dt1) edge[bend left,line width=2pt] node[yshift=0mm] {$a$} (dt2);
      ;
    \end{tikzpicture}}
  \caption{Abstract sketch of the determinization guided by SCCs of $\op{PS}(\mc{A})$.
  The two powerset SCCs $C$ and $C'$ were determinized starting with the sets $S$ and $S'$,
  respectively. The states in the constructed partial DPAs $\mc{B}_C$ and $\mc{B}_{C'}$
  are depicted by the sets from $\op{PS}(\mc{A})$ to which they are language-equivalent.
  It suffices to keep just one bottom SCC of each partial DPA (e.g. in the dotted
  rectangle) and connect them by exploiting the known language equivalences to introduce
  missing edges (e.g. the depicted bold edge).
  }
  \label{fig:toposketch}
\end{figure}

Putting this into practice and keeping only bottom SCCs of the partial DPAs that have
the smallest size, is what we call the ``topological'' optimization. While this still requires
the exploration of all macrostates that are reached by the determinization of a single SCC
in $\op{PS}(A)$, many of those may be removed afterwards.

Alternatively, one could explore
in a depth-first fashion and greedily just keep the first completed bottom SCC,
effectively trading an even smaller automaton size for possibly faster computation.

It can be easily implemented in such a way that the resulting automaton will have at most
the same size as without this optimization, by picking appropriate initial states for the
construction of the partial DPAs $\mc{B}_C$ (i.e. picking initial states that would also
be reached by the unoptimized construction anyway).

Notice that this optimization is generic and can be used with any determinization
construction based on state exploration.

\subsection{State reduction using Mealy minimization}
\label{subsec:minimize}

A simple but powerful optimization that to our knowledge has not been considered yet,
consists of first minimizing the number of priorities of the DPA (which can be done
efficiently using a simple algorithm from \cite{carton1999computing}) and then applying
the classical minimization algorithm by Hopcroft \cite{hopcroft1971n}, by interpreting the
DPA as a Mealy machine that outputs priorities. Since the minimization as Mealy machine
preserves the priority sequence for each input, this clearly does not change the language
of the DPA.

We believe that the effectiveness is due to
the determinization procedure often returning automata with an unnecessarily large number
of different priorities. Thus, the new priority function may assign edges with initially
distinct priorities the same new priority. Thereby, paths with initially different
priority sequences might become equivalent, ultimately leading to more states that can be
merged.
Because this is a pure post-processing step, it can be integrated
easily into any determinization tool-chain.
 \section{Experiments and Discussion}
\label{sec:experiments}

Our quantitative experiments\footnote{Our prototype can be obtained at \url{https://github.com/apirogov/nbautils}} involved different sets of Büchi automata that can be roughly
grouped into two categories---automata obtained from LTL formulas and arbitrary automata.
Most sets were generated and/or processed using tools included with \texttt{spot}
\cite{duret2016spot} (version 2.6.3). All sets were filtered to exclude already
deterministic Büchi automata. Furthermore, we filtered the sets to include only automata
that are determinizable by \texttt{spot} (via \texttt{autfilt -D -P -{}-high}) in
reasonable time (between 1-30 minutes, depending on the size of the corresponding test
set). This was also our benchmark for comparison. We explicitly compare only the size of
the resulting automata and do not aim for competitive performance.
However, most heuristics we proposed are computationally cheap,
being essentially simple modifications of the successor calculation.

All LTL formulas were transformed to state-based Büchi automata using \texttt{ltl2tgba
-B}. As we are evaluating only the impact of the heuristics during the transformation from
Büchi to parity automata, the method to transform LTL formulas to Büchi automata does not
matter. For the same reason, we do not compare with tools like \texttt{Rabinizer}
\cite{kvretinsky2018rabinizer} that bypass Büchi automata and Safra-style constructions.
We are also aware of the fact that \texttt{spot} might have been able to produce smaller
automata from LTL when using the usually smaller transition-based automata as input. As
all our heuristics and optimizations in principle would work on those as well, we believe
that this does not matter for our evaluation.

We used \texttt{autcross} to obtain statistics about different combinations of our
heuristics as well as to ensure correctness of the generated automata.
The test sets that we used were the following:
\begin{itemize}
  \item[\textbf{raut}:] 1662 random automata generated by \texttt{randaut} with
    5-20 states (9 on average), mostly in one SCC.

  \item[\textbf{rautms}:] 2112 random automata with 10-30 states (13 on average)
    generated by \texttt{randaut}, but filtered to have more than one SCC, specifically
    having at least one rejecting and at least one accepting component.

  \item[\textbf{rltl}:] 11798 automata of various size (8 states on average) from random
    LTL formulas generated by \texttt{randltl}.

  \item[\textbf{gltl}:] 68 automata from various parametrized LTL formula families
    generated by \texttt{genltl} that yield nondeterministic automata
    \cite{gastin2001fast,kupferman2005linear,geldenhuys2006larger,kupferman2010blow,tabakov2010optimized,duret2013manipulating,mullersickertformula}.

  \item[\textbf{scomp}:] 113 automata from LTL formulas extracted from the TLSF-benchmarks
    directory of the SYNTCOMP \cite{jacobs20174th} benchmark repository.

\end{itemize}

To reduce the exponential number of possible combinations of different heuristics and
algorithm variants to a manageable amount, we fixed an order in which to add optimizations
additionally to the previous configuration. We tested this sequence of cumulatively
enabled optimizations for the three successor construction strategies proposed in the
description of the generalized determinization algorithm \cite{unidet}, which correspond
to the Muller-Schupp construction (abbreviated M.-S.), Safra's construction and the
maximally collapsing merge rule (abbreviated Max.). The following list explains the used
parameter combinations.
In the tables, the optimization listed in a column is enabled \emph{in addition} to the
ones used in all previous columns.
\begin{itemize}
  \item[\textbf{def}:] only trimming of the automaton and the known true-loop optimization \cite{klein2005diploma}
  \item[\textbf{+T}:] topological optimization (Sec.~\ref{subsec:topological})
  \item[\textbf{+E}:] known external language inclusion \cite{duret2016spot} (Sec.~\ref{subsec:langinc})
  \item[\textbf{+I}:] internal language inclusion (Sec.~\ref{subsec:langinc})
  \item[\textbf{+M}:] priority minimization and state reduction (Sec.~\ref{subsec:minimize})
  \item[\textbf{+S}:] smart successor selection using tries (Sec.~\ref{subsec:smart})
  \item[\textbf{+A}:] breakpoint construction for accepting SCCs (Sec.~\ref{subsec:NBASCCs})
  \item[\textbf{+W}:] optimization for both, accepting and rejecting SCCs (Sec.~\ref{subsec:NBASCCs})
  \item[\textbf{+D}:] optimization for deterministic SCCs (Sec.~\ref{subsec:NBASCCs})
\end{itemize}
The results for our main benchmark test sets are shown in Table~\ref{tab:experiments}.

\begin{table}[h]
  \def\arraystretch{1.2}
  \setlength{\tabcolsep}{3pt}
  \begin{center}
  \begin{tabular}{r|c|l|llllllll}
    & $\Sigma$ \texttt{spot} & mode & def & +T & +E & +I & +M & +S & +W & +D \\\hline
    & & M.-S. & 2.42 & 2.31 & 2.31 & 2.31 & 1.33 & 1.07 & 1.07 & 1.07 \\
    \textbf{raut} & 334523 & Safra & 1.14 & 1.00 & 1.00 & 1.00 & \textbf{0.80} & 0.82 & 0.82 & 0.82 \\
    & & Max.  & 1.21 & 1.07 & 1.06 & 1.06 & 0.84 & 0.83 & 0.83 & 0.83 \\
    \hline
    & & M.-S. & 1.28 & 1.21 & 0.92 & 0.92 & 0.77 & 0.70 & 0.68 & 0.68 \\
    \textbf{rautms} & 286967 & Safra & 0.70 & 0.70 & 0.68 & 0.68 & 0.62 & 0.61 & \textbf{0.60} & 0.60 \\
    & & Max.  & 0.73 & 0.72 & 0.71 & 0.71 & 0.64 & 0.62 & 0.61 & 0.61\\
    \hline
    & & M.-S. & 3.63 & 2.05 & 1.42 & 1.33 & \textbf{0.84} & 0.93 & 0.94 & 0.90 \\
    \textbf{rltl} & 174251 & Safra & 4.15 & 2.67 & 1.76 & 1.61 & 0.93 & 0.96 & 0.97 & 0.92 \\
    & & Max.  & 3.94 & 2.47 & 1.79 & 1.65 & 0.98 & 0.99 & 1.01 & 0.93 \\
    \hline
    & & M.-S. & 3.12 & 1.98 & 1.97 & 1.87 & 1.18 & 1.08 & 1.07 & 1.00 \\
    \textbf{gltl} & 3658 &  Safra & 2.77 & 1.93 & 1.92 & 1.84 & 1.18 & 1.07 & 1.06 & 0.995 \\
    & & Max.  & 2.39 & 1.71 & 1.70 & 1.62 & 1.00 & 1.05 & 1.05 & \textbf{0.994} \\
    \hline
    & & M.-S. & & & 1.84 & 1.59 & 0.97  & 0.77 & 0.85 & 0.74 \\
    \textbf{scomp} & 14502 & Safra & & & 1.92 & 1.58 & 0.98 & 0.72 & 0.78 & \textbf{0.69} \\
    & & Max.  & & & 1.62 & 1.22 & 0.73 & 0.74 & 0.85 & 0.70 \\
  \end{tabular}
  \end{center}

  \caption{Results of the quantitative experiments. The
  numbers show the fraction of the sum of states obtained by the indicated configuration
  of our implementation, compared to the sum of states of the automata obtained by
  \texttt{spot}. On the SYNTCOMP set we did not evaluate the least optimized
  configurations, as some optimizations were necessary to obtain results in acceptable time.
  }
  \label{tab:experiments}
\end{table}

Our results show that for each test set there is a configuration of our prototype that,
on average, produces automata that are not larger than the ones produced by \texttt{spot}.
In many cases, the best configuration produces automata with up to $40\%$ less states.

Unfortunately, the effect of the post-processing step with the priority minimization
and state reduction (\textbf{+M}) is not robust even under minor variations of
the automaton. This is, e.g., witnessed by the test set \textbf{rltl},
for which the state reduction without the smart successor selection yields
better results.  As only states that agree on all priority sequences
of runs going through them can be merged by Hopcroft's algorithm,
slight variations that lead to ``symmetry breaking'' may already
render states non-mergeable. This optimization is more
costly than the others, because equivalent states must be computed in
the resulting automaton. But we believe that this is one of the most
generally useful optimizations which is worth the additional
computation time, as it quite consistently reduces the number of
states by approximately $20$--$40\%$.

The smart selection of successors (\textbf{+S}) is the only proposed optimization which is ``stateful'',
in the sense that it depends on the already constructed part of the automaton, which in
turn depends on other enabled optimizations. In all but one test sets, this optimization
was especially helpful in reducing states when using the Muller-Schupp successors.
Interestingly, it was most successful in the Safra-based construction for the SYNTCOMP
test set, where this optimization alone was responsible for an additional state reduction
of $26\%$.
The topological optimization (\textbf{+T}) yields a mild state reduction on random automata and
significant reduction of up to $36\%$ in the LTL-based test sets.
\withappendix{
A particularly suggestive demonstration of the topological optimization is given in
Appendix~\ref{app:experiments}, Figure~\ref{fig:goodtopo}.
}

The application of the optimization that exploits language inclusion relations inside of a
single SCC (\textbf{+I}) seems to have no visible effect on random automata, but was responsible for a
decent additional state reduction of $6$--$8\%$ for any variant of the construction in
the LTL-based test sets. In our prototype we used only a simple direct simulation
\cite{etessami2000optimizing} for the optimizations. We are optimistic that using more
involved simulations \cite{etessami2005fair} to under-approximate language inclusion, e.g.
fair simulation, would lead to even better results.

The heuristics for accepting and rejecting SCCs of the NBA (\textbf{+W}), while
sometimes helpful, in other cases lead to an increase of the number of
states. The separate handling of deterministic components, which is
enabled in addition to all other optimizations, shows a positive
effect mainly in the LTL-based automata test sets.

We now give some further examples that illustrate the potential effect of some of the heuristics.
We start with the effect of our state reduction when determinizing the family of
automata introduced by Max Michel in \cite{michelreport} (see also \cite{Thomas97}) and which was used as a benchmark in
\cite{althoff2005observations}. A deterministic automaton recognizing the same language must
have at least $n!$ states and is usually much larger in practice. While in
\cite{althoff2005observations} the Muller-Schupp construction performed significantly worse
than Safra's construction, we were surprised to see that using the Muller-Schupp update in
our own experiments, the Hopcroft minimization is able to drastically reduce the size of the
automaton, producing much smaller automata than any other construction variant we tried
(see Table~\ref{tab:michel}).
\begin{table}[h]
  \begin{center}
  \def\arraystretch{1.2}
  \setlength{\tabcolsep}{3pt}
  \begin{tabular}{c|r|l|rr||c|r|l|rr}
     $n$ & \texttt{spot} & mode & def+TEI & +M
   & $n$ & \texttt{spot} & mode & def+TEI & +M
    \\
    \hline
     2 & 18    & M.-S. & 19  & \textbf{17}
     & 4 & 2284  & M.-S. & 2725 & \textbf{457}
     \\
       &       & Safra & 18  & 18
       &   &       & Safra & 2202 & 2106
     \\
       &       & Max.  & 18  & 18
       &   &       & Max.  & 2094 & 2094
       \\
    \hline
     3 & 145   & M.-S. & 166 & \textbf{82}
     & 5 & 60109 & M.-S. & 72616 & \textbf{2936}
     \\
       &       & Safra & 142 & 142
       &   &       & Safra & 57714 & 57714
       \\
       &       & Max.  & 142 & 142
       &   &       & Max.  & 51094 & 51094
       \\
  \end{tabular}
  \end{center}
  \vspace{-2mm}
  \caption{Results for Michel's automata family (input NBA have $n+1$ states for all $n$).
  The table demonstrates the surprising efficiency of the Muller-Schupp construction with
  subsequent state reduction.
  }
  \vspace{-10mm}
  \label{tab:michel}
\end{table}

The heuristics for handling rejecting/accepting (\textbf{+W}) and deterministic (\textbf{+D})
SCCs of the NBA were enabled in addition to all other heuristics in the
experiments shown in Table~\ref{tab:experiments}, and thus only had a
small or even negative effect.
\withappendix{
  We refer to Appendix~\ref{app:experiments}, Figure~\ref{fig:badforheu} for a concrete example
where SCC-based heuristics can be harmful.
}
In the following we give some examples showing that
these heuristics indeed can have a strong positive effect. The deterministic
SCC optimization is very useful for determinization of automata obtained from
formulas $\varphi_{GH}(n) := \bigwedge_{i=1}^n \textsf{G}\textsf{F}a_i
\lor \textsf{F}\textsf{G}a_{i+1}$ from \cite{geldenhuys2006larger}.
An example where the heuristic which applies the breakpoint-style
construction to accepting SCCs is very helpful is the family of
formulas $\varphi_{MS}(n) := \bigvee_{i=0}^{n}
\textsf{F}\textsf{G}(\lnot^{i} a\lor \textsf{X}^{i} b)$ from
\cite{mullersickertformula}, where the bottom SCCs of the automata are increasingly
complex accepting SCCs (see Table~\ref{tab:goodforheu}).

\begin{table}[h]
  \begin{center}
  \def\arraystretch{1.2}
  \setlength{\tabcolsep}{3pt}
    \begin{tabular}{r|r|rr|rr|rr}
              &               & M.-S. & & Safra & & max. & \\\hline
      formula & \texttt{spot} & def.+TEI & +D & def.+TEI & +D & def.+TEI & +D \\
      \hline
      $\varphi_{GH}(2)$ & 18    & 30 & 32 & 29 & 32 & 28 & 32 \\
      $\varphi_{GH}(3)$ & 108   & 385 & 255 & 327 & 255 & 381 & 266 \\
      $\varphi_{GH}(4)$ & 2143  & 15206 & 5612 & 10922 & 5612 & 12394 & 5036 \\
      \hline
      formula & \texttt{spot} & def.+TEI & +A & def.+TEI & +A & def.+TEI & +A \\
      \hline
      $\varphi_{MS}(2)$ & 21    & 40 & 16 & 40 & 16 & 36 & 16 \\
      $\varphi_{MS}(3)$ & 170   & 371 & 46 & 371 & 46 & 155 & 46 \\
      $\varphi_{MS}(4)$ & 1816  & 4933 & 132 & 4933 & 132 & 620 & 132 \\
      $\varphi_{MS}(5)$ & 22196 & 67173 & 358 & 67173 & 358 & 2419 & 358 \\
    \end{tabular}
  \end{center}

  \caption{
  Results for some instances of the $\varphi_{GH}(n)$ and $\varphi_{MS}(n)$ formulas, with and
  without usage of suitable heuristics.
  Number of states obtained by \texttt{spot} is provided as reference in the tables.
  In case of $\varphi_{GH}$ formulas, $\texttt{spot}$ profits significantly from the
  stutter-invariance optimization \cite{klein2007fly}, which we do not utilize.
  }
  \vspace{-10mm}
  \label{tab:goodforheu}
\end{table}

In summary, we noticed that the positive effects of the heuristics become stronger with growing
size of the input automaton. This is not surprising, as for e.g. the SCC heuristics to
have a positive effect, the automaton needs to be sufficiently complex. Unfortunately,
larger input automata are not suitable for a thorough quantitative analysis in reasonable
time. We are convinced, that for sufficiently large automata, even the heuristics that may
appear not very effective in the presented benchmarks would have a stronger positive
effect on average.

Our results also show that every proposed choice of the $\op{merge}$ operation in the
unified determinization construction seems to be superior in some cases.
The Muller-Schupp update was the best choice on the random LTL test
set, while the maximal collapse update was superior on the test set with the parametrized
LTL formulas. Even though the unoptimized Muller-Schupp update usually seems to perform worst,
combining it with some optimizations makes it a viable choice. The maximal collapsing
update seems to perform comparably with the Safra update, but appears to be slightly worse
in most cases, whereas the well-known Safra update seems to be a good middle ground.

 \vspace{-2mm}
\section{Conclusion}
\label{sec:conclusion}

\vspace{-2mm}
We presented a number of new heuristic optimizations for
the determinization of Büchi automata, and evaluated them in a
prototype implementation using different test sets of automata,
ranging from randomly generated automata to automata constructed from
specifications from the competition SYNTCOMP. Our results show that
these heuristics can significantly reduce the number of states in
comparison with the base construction, and also in
comparison with the current state-of-the-art tool \texttt{spot} for
determinization of NBAs.

In future work, we want to study in more depth the effect of the
heuristic based on language inclusions, by using stronger tools for
identifying language inclusions between states of the Büchi automaton
(currently we are only using direct simulation). We also see further
potential in the smart successor selection, which could be used for
redirecting transitions in the constructed automaton in order to
reduce its size.

\bibliographystyle{splncs04}
\bibliography{literature}

\begin{thebibliography}{10}
\providecommand{\url}[1]{\texttt{#1}}
\providecommand{\urlprefix}{URL }
\providecommand{\doi}[1]{https://doi.org/#1}

\bibitem{althoff2005observations}
Althoff, C.S., Thomas, W., Wallmeier, N.: Observations on determinization of
  {B}{\"u}chi automata. In: CIAA 2005. pp. 262--272. Springer

\bibitem{baier2008principles}
Baier, C., Katoen, J.P.: Principles of model checking. MIT Press (2008)

\bibitem{boigelot2001use}
Boigelot, B., Jodogne, S., Wolper, P.: On the use of weak automata for deciding
  linear arithmetic with integer and real variables. In: IJCAR 2001. Springer

\bibitem{buchi1966symposium}
B{\"u}chi, J.R.: On a decision method in restricted second order arithmetic.
  In: Studies in Logic and the Foundations of Mathematics, vol.~44, pp. 1--11.
  Elsevier (1966)

\bibitem{carton1999computing}
Carton, O., Maceiras, R.: Computing the {R}abin index of a parity automaton.
  RAIRO-Theoretical Informatics and Applications  \textbf{33}(6),  495--505
  (1999)

\bibitem{ChoffrutG99}
Choffrut, C., Grigorieff, S.: Uniformization of rational relations. In: Jewels
  are Forever. pp. 59--71. Springer (1999)

\bibitem{ColcombetZ09}
Colcombet, T., Zdanowski, K.: A tight lower bound for determinization of
  transition labeled {B}{\"u}chi automata. In: ICALP 2009. Springer

\bibitem{duret2013manipulating}
Duret-Lutz, A.: Manipulating {LTL} formulas using spot 1.0. In: ATVA 2013.
  Springer

\bibitem{duret2016spot}
Duret-Lutz, A., Lewkowicz, A., Fauchille, A., Michaud, T., Renault, E., Xu, L.:
  Spot 2.0--a framework for {LTL} and $\omega$-automata manipulation. In: ATVA
  2016. Springer

\bibitem{esparza2017ltl}
Esparza, J., K{\v{r}}et{\'\i}nsk{\`y}, J., Raskin, J.F., Sickert, S.: From
  {LTL} and limit-deterministic {B}{\"u}chi automata to deterministic parity
  automata. In: TACAS 2017. Springer

\bibitem{etessami2000optimizing}
Etessami, K., Holzmann, G.J.: Optimizing {B}{\"u}chi automata. In: CONCUR 2000.
  Springer

\bibitem{etessami2005fair}
Etessami, K., Wilke, T., Schuller, R.A.: Fair simulation relations, parity
  games, and state space reduction for {B}{\"u}chi automata. SIAM J. on
  Computing  \textbf{34}(5) (2005)

\bibitem{FismanL15}
Fisman, D., Lustig, Y.: A modular approach for {B}{\"{u}}chi determinization.
  In: {CONCUR} 2015. LIPIcs

\bibitem{fogarty2015profile}
Fogarty, S., Kupferman, O., Vardi, M.Y., Wilke, T.: Profile trees for
  {B}{\"u}chi word automata, with application to determinization. Information
  and Computation  \textbf{245},  136--151 (2015)

\bibitem{gastin2001fast}
Gastin, P., Oddoux, D.: Fast {LTL} to {B}{\"u}chi automata translation. In: CAV
  2001. Springer

\bibitem{geldenhuys2006larger}
Geldenhuys, J., Hansen, H.: Larger automata and less work for {LTL} model
  checking. In: SPIN 2006. Springer

\bibitem{hopcroft1971n}
Hopcroft, J.: An n log n algorithm for minimizing states in a finite automaton.
  In: Theory of machines and computations, pp. 189--196. Elsevier (1971)

\bibitem{jacobs20174th}
Jacobs, S., Basset, N., Bloem, R., Brenguier, R., Colange, M., Faymonville, P.,
  Finkbeiner, B., Khalimov, A., Klein, F., Michaud, T., et~al.: The 4th
  reactive synthesis competition (syntcomp 2017): Benchmarks, participants \&
  results. arXiv preprint arXiv:1711.11439  (2017)

\bibitem{kahler2008complementation}
K{\"a}hler, D., Wilke, T.: Complementation, disambiguation, and determinization
  of {B}{\"u}chi automata unified. In: ICALP 2008. Springer

\bibitem{klein2005diploma}
Klein, J.: Linear time logic and deterministic omega-automata. Diploma thesis,
  University of Bonn (2005)

\bibitem{klein2007fly}
Klein, J., Baier, C.: On-the-fly stuttering in the construction of
  deterministic omega-automata. In: CIAA 2007. Springer

\bibitem{kvretinsky2018rabinizer}
K{\v{r}}et{\'\i}nsk{\`y}, J., Meggendorfer, T., Sickert, S., Ziegler, C.:
  Rabinizer 4: from {LTL} to your favourite deterministic automaton. In: CAV
  2018. Springer

\bibitem{kupferman2010blow}
Kupferman, O., Rosenberg, A.: The blow-up in translating {LTL} to deterministic
  automata. In: MoChArt 2010. Springer

\bibitem{kupferman2005linear}
Kupferman, O., Vardi, M.Y.: From linear time to branching time. TOCL 2005

\bibitem{unidet}
L\"oding, C., Pirogov, A.: Determinization of {B\"uchi} automata: Unifying the
  approaches of {Safra} and {Muller}-{Schupp}. ICALP 2019 (to appear)
  \url{https://arxiv.org/abs/1902.02139}

\bibitem{mcnaughton1966testing}
McNaughton, R.: Testing and generating infinite sequences by a finite
  automaton. Information and control  \textbf{9}(5),  521--530 (1966)

\bibitem{MeyerSL18}
Meyer, P.J., Sickert, S., Luttenberger, M.: Strix: Explicit reactive synthesis
  strikes back! In: CAV 2018. Springer

\bibitem{michelreport}
Michel, M.: Complementation is more difficult with automata on infinite words.
  Manuscript, CNET, Paris  (1988)

\bibitem{miyano1984alternating}
Miyano, S., Hayashi, T.: Alternating finite automata on $\omega$-words.
  Theoretical Computer Science  \textbf{32}(3),  321--330 (1984)

\bibitem{mullersickertformula}
M{\"{u}}ller, D., Sickert, S.: {LTL} to deterministic {E}merson-{L}ei automata.
  In: GandALF 2017

\bibitem{muller1995simulating}
Muller, D.E., Schupp, P.E.: Simulating alternating tree automata by
  nondeterministic automata: New results and new proofs of the theorems of
  {R}abin, {McNaughton} and {S}afra. Theoretical Computer Science
  \textbf{141}(1-2),  69--107 (1995)

\bibitem{piterman2006nondeterministic}
Piterman, N.: From nondeterministic {B}{\"u}chi and {S}treett automata to
  deterministic parity automata. In: LICS 2006. IEEE

\bibitem{pnueli1977temporal}
Pnueli, A.: The temporal logic of programs. In: Foundations of Computer
  Science, 1977., 18th Annual Symposium on. pp. 46--57. IEEE (1977)

\bibitem{redziejowski2012improved}
Redziejowski, R.R.: An improved construction of deterministic omega-automaton
  using derivatives. Fundamenta Informaticae  \textbf{119}(3-4),  393--406
  (2012)

\bibitem{safra1988complexity}
Safra, S.: On the complexity of omega-automata. In: Foundations of Computer
  Science, 1988., 29th Annual Symposium on. pp. 319--327. IEEE (1988)

\bibitem{schewe2009tighter}
Schewe, S.: Tighter bounds for the determinisation of {B}{\"u}chi automata. In:
  FOSSACS 2009. Springer

\bibitem{tabakov2010optimized}
Tabakov, D., Vardi, M.Y.: Optimized temporal monitors for {SystemC}. In:
  International Conference on Runtime Verification. pp. 436--451. Springer
  (2010)

\bibitem{Thomas90}
Thomas, W.: Automata on infinite objects. In: Handbook of Theoretical Computer
  Science, vol.~B, pp. 133--192. Elsevier Science Publishers, Amsterdam (1990)

\bibitem{Thomas97}
Thomas, W.: Handbook of formal languages, vol. 3. chap. Languages, Automata,
  and Logic, pp. 389--455. Springer (1997)

\bibitem{Thomas08}
Thomas, W.: Church's problem and a tour through automata theory. In: Pillars of
  Computer Science. pp. 635--655. Springer (2008)

\bibitem{VardiW07}
Vardi, M.Y., Wilke, T.: Automata: from logics to algorithms. In: Logic and
  automata - history and perspectives, Texts in Logic and Games, vol.~2, pp.
  629--724. Amsterdam University Press (2007)

\end{thebibliography}

\withappendix{
\newpage
\appendix
\section{Adaptations for the modularized construction}
\label{app:modular}

\vspace{-3mm}

In this section, we provide some more details on our modularized determinization construction
that is based on the construction presented in \cite{unidet} and was sketched in
Section~\ref{sec:modularsketch}.

Fix a Büchi automaton $\mc{A}=(\Sigma,Q,\Delta,Q_0,F)$ and observe that each accepting
run of the Büchi automaton must eventually stay
in the same SCC of $\mc{A}$ forever. Furthermore, a run can switch between components
only a finite number of times. So one can run multiple instances of the
determinization in parallel, each responsible for one (or multiple) of the SCCs of the
Büchi automaton, and try to detect accepting runs that eventually stay in those components.

Formally, let $P = \{ P_i \}_{i=1}^{|P|}$ be an SCC-preserving partition of $Q$,
where SCC-preserving means that for $p,q\in Q$, $\scc(p) = \scc(q)$ implies that $p,q \in
P_i$ for the same $i\in \Nat$.
Fix $m\in \Nat$, denoting the number of determinization tuples that will be run in
parallel, and choose an injective map $\op{prt} : [m] \hookrightarrow P$ that assigns each
tuple a different set of states from the Büchi automaton that it will be responsible for.
The assignment should cover all SCCs of the Büchi automaton which admit accepting runs, i.e.
$Q \setminus R \subseteq \bigcup_{i=1}^m \op{prt}(i)$, where $R$ is the set of
states in rejecting SCCs of $\mc{A}$.

A macrostate $s := ((\alpha_i, t_i)_{i=1}^m, B)$ in the modular construction consists of a
\emph{buffer set} $B\subseteq Q$ and \emph{components} $(\alpha_i,t_i)$ with tuples
$t_i := (S_{i,1}, …, S_{i,n_i})$ equipped with injective functions $\alpha_i : [n_i] \hookrightarrow [n]$
that map each set to its rank, where $n := \sum_{i=1}^m n_i$, $Q_{t_i} := \bigcup_{j=1}^{n_i} S_{i,j}$
and $Q_s := B \cup \bigcup_{i=1}^m Q_{t_i}$.
We require that the sets $\op{img}(\alpha_i)$ partition $[n]$,
$Q_{t_i} \subseteq \op{prt}(i)$ and that all sets of states in $s$ must be disjoint.
We allow the components to be empty, i.e. having $n_i = 0$ (then $\alpha_i$ is just the
empty function).

Whenever a state (of the NBA) is reached from some other state
in the same component, we must follow the
rules of the original procedure to infer accepting runs correctly. However, when
a state is reached only from a state in a different component, then the state must lie on a
run which is not already represented in the tuple for that component. Hence, we can add
this state to the component in such a way that it is not related to
the others and from there on let the original construction handle the rest.

The priorities at the transitions of the original construction are consequence of red and
green ranks that are assigned to sets. When multiple components work in
parallel and one component witnesses an accepting run by signalling green infinitely
often (and red only finitely often), these signals must eventually dominate red signals
from other components. The correct interaction of signals is easily achieved by using a
ranking that totally orders all the sets (which is captured by the definition and
constraints on partial rankings $\alpha_i$), disregarding the separated components, and
updating the ranks according to the rules of the construction. Following the construction
ensures that ``less important'' sets move up in the ranking order whenever a more
important rank has a bad event.

The initial state consists of a macrostate $q_0$ such that all states in $Q_0$ are
distributed into singleton sets in the corresponding tuples $t_i$, with arbitrary
different ranks satisfying the definition of macrostates above.

A transition is defined as follows. Let $s=((\alpha_i, t_i)_{i=1}^m, B)$ be some macrostate.
Similar to the original construction, the successor on some letter $x\in \Sigma$ is obtained by
performing the steps $\op{step}, \op{relocate}, \op{prune}, \op{merge}$ and
$\op{normalize}$, where only the $\op{relocate}$ step is completely new, whereas the other
steps are only small modifications of the original construction.

First, the operation $\op{step}$ is executed on each tuple exactly as before.
Additionally, let $\tilde{B} := \Delta(B, x)$ be the successors of states in the buffer set.
The new operation $\op{relocate}$ collects the set of all wrongly-placed states
from the sets $\tilde{B}$ and $\tilde{Q}_{t_i}$, removes them from those sets and places
them into the correct sets.

Formally, let $W_0 := \tilde{B} \cap \bigcup_{i=1}^m \op{prt}(i)$ and for $1\leq i
\leq m$ let $W_i := \{ q \in Q_{\tilde{t}_i} \mid q \not\in \op{prt}(i) \}$.
Then $\mathring{B} := (\tilde{B} \setminus W_0) \cup ((\bigcup_{i=1}^m W_i) \setminus
\bigcup_{i=1}^m \op{prt}(i))$ is the new buffer set, containing all reached states that do
not belong in any tuple. The tuples $\mathring{t}_i$ are defined by setting
$\mathring{S}_{i,j} := \tilde{S}_{i,j} \setminus W_i$ for each set
$\tilde{S}_{i,j}$ in $\tilde{t}_i$, thereby removing wrongly placed states, and then
adding a new rightmost set $\mathring{S}_{i,new} :=  (\bigcup_{j=0}^m W_j \cap
\op{prt}(i)) \setminus Q_{\tilde{t}_i}$ to the tuple, which contains states that have
predecessors only outside of the tuple. The ranks $\mathring{\alpha}$ for
the modified sets stay the same as before, whereas all the new sets $\mathring{S}_{i,new}$
are assigned some fresh rank higher than all currently assigned ranks, giving us the
intermediate macrostate $((\mathring{\alpha}_i,\mathring{t}_i)_{i=1}^m, \mathring{B})$.

Next, proceed with $\op{prune}$ as in the original construction.
The priority of the transition is obtained in essentially the same way as before, based on
the results of $\op{prune}$ in the tuples. First obtain the sets $G_i$ and $R_i$ of good
and bad ranks of the different tuples according to the original construction, then
consider their union $G:=\bigcup_{i=1}^m G_i$ and $R:=\bigcup_{i=1}^m R_i$. The priority $p$ of the
transition is $2k$ if $k\in G$ and $2k - 1$ otherwise, where $k := \min (G\cupdot R)$ (or
$k := |Q|+1$ if $(G\cupdot R)=\emptyset$) denotes the smallest rank associated with some
event in a tuple.

The $\op{merge}$ operation has constraints based on the smallest rank with a good or bad
event in the transition. Use the most important rank $k$ of the whole macrostate from
above to execute $\op{merge}$ on the individual tuples.
Finally, apply $\op{normalize}$ in such a way that the rank order of all sets in the whole
macrostate is preserved and let $B' := \mathring{B}$, so that in the end we have the
successor macrostate $s'=((\alpha'_i, t'_i)_{i=1}^m, B')$.

A note to readers familiar with Safra's construction---the redistribution of misplaced states
in the $\op{relocate}$ operation essentially introduces a new separate Safra tree root
labelled by the newly introduced states. The tree structure here is encoded in the ranks (more
details on this correspondence can be found in \cite{unidet}). By adding the new rightmost
sets $\mathring{S}_{i,new}$ as described above, each tuple in general contains not just one tree
encoding the relationship of runs that share a common history, but multiple such trees in
an ordered forest, reflecting the ``unknown'' history which was lost in the transition of
the state between the different tuples. The important preserved information is that the
state was actually reached at some time $i$ and from this point on the runs that continue
from there and stay in the same component will be correctly managed like in the original
construction.

 \newpage
\section{Correctness of internal language inclusion optimization}
\label{app:inclusions}

In this section we give a proof sketch for the proposed internal language inclusion
optimization from Section~\ref{subsec:langinc}:

\thmlanginclusions*

\begin{proof}[sketch]
  Let $w\in L(\mc{A})$. We need to show that removing states in this way does not inhibit
  the acceptance of $w$ by the resulting DPA.
  Let the acceptance profile of a run $q_0 q_1 \ldots$ be defined as a sequence
  $z_0 z_1 \ldots \in \{0,1\}^\omega$ such that $z_i = 1$ if $q_i \in F$ and $z_i=0$ otherwise.
  Then let $\rho$ be an accepting run with a lexicographically maximal
  acceptance profile, i.e., an accepting run visiting accepting states
  as early as possible (such a run always exists, see e.g., the greedy
  ordering of runs in \cite{ChoffrutG99}).

  The determinization construction puts states along lexicographically
  larger runs to the left of states along lexicographically smaller
  runs by separating accepting from non-accepting states and keeping
  only leftmost occurrences of each state in $\op{step}$. This order is
  never swapped during the other operations. It can only happen that
  two states are merged into the same set in $\op{merge}$.
  Thus, the construction ensures the following invariant: $\op{idx}_t(q) < \op{idx}_t(p)$ implies
  that a lex.\ maximal run going through $q$ is lex.\ larger than a lex.\ maximal run
  going through $p$ at the same time.
  Hence, in the sequence of macrostates on $w$ in a DPA, no state
  along the run $\rho$ is ever removed by this optimization, because this would contradict
  the maximality of the acceptance profile of $\rho$. Choosing this run $\rho$ in the
  proof of \cite[Lemma 6]{unidet}
  shows that the DPA indeed still accepts $w$. \qed
\end{proof}

\subsection{Modification for the modular construction}

The invariant in the proof above may be violated in the modular construction
(Appendix~\ref{app:modular}) due to states being moved between different tuples, such that
some run which goes through these states may suddenly switch to the right side of a tuple,
even though it is ``better'' than other runs tracked in that tuple. This poses a problem,
if the same run can leave and re-enter the same tuple. Therefore, we need to add an
additional constraint to pairs of states for which we can apply this optimization to
exclude this possibility.

More formally, in the modular construction we may only apply Theorem~\ref{thm:inclusions}
to some pair of NBA states $p,q \in Q$ with $L(p) \subseteq L(q)$ that are assigned to the same
determinization tuple $t_i$ (i.e., $p,q \in \op{prt}(i)$ for some $i\in\Nat$), if every path
from $q$ to $p$ must stay in $\op{prt}(i)$, i.e., can not leave and re-enter $t_i$.
 \section{Details on the accepting SCC heuristic}
\label{app:accsccs}

In this section we describe how the breakpoint construction \cite{miyano1984alternating}
can be internalized in the determinization procedure to handle accepting SCCs of Büchi
automata more efficiently.

The idea of the breakpoint construction is that for accepting SCCs it is sufficient to
check that at least one run eventually stays in the same SCC forever. Because there are no
rejecting loops that could be taken by runs in these SCCs, it suffices to (conceptually)
fix some set of runs that are currently in such SCCs and track sets of states that are
reached by these runs and stay in those SCCs. As long as there is at least one state
left, there must be a run remaining in this SCC and all remaining runs must be accepting.
This reduces the required effort to track accepting runs significantly. One only needs to
ensure the separation of the states in these accepting SCCs reached by the currently
tracked runs from states that were reached only by runs that are currently not followed.
When our currently tracked set becomes empty, we can start the tracking again on (a subset
of) states along the previously unfollowed runs. If it is ensured that every reached state
in the accepting SCCs eventually has the chance to be tracked, an accepting run will
eventually be detected, if one exists.

So we need to manage only two different sets to detect accepting runs in any accepting
SCC---a \emph{track set} and a \emph{background set}. In a transition, the track set is
updated to all successors of the current track set that are in an accepting SCC. The
background set contains all other states from accepting SCCs that are reached in this
transition. Whenever the track set becomes empty (which is a ``breakpoint''), all states
from the background set are moved to the track set. Using this kind of scheme (see e.g.
\cite{boigelot2001use}), it is possible to determinize weak NBA with a blow-up of at most
$3^{n}$, instead of the general upper bound of $2^{\mc{O}(n\log n)}$.

To exploit this approach for general NBA, we utilize our modular determinization
construction (see Appendix~\ref{app:modular}). We use the buffer set $B$ of the
modular macrostates for the background set and a separate determinization component for the track
set, which is formally some ranked tuple $(\alpha_i,t_i)$ with $t_i := (S_{i,1})$
containing a single set with all the states in the track set. This component is
responsible for all accepting SCCs, i.e., $\op{prt}(i) := A$, where $A$ is the union of
all accepting SCCs of the NBA. The operation $\op{relocate}$ is modified in such a way
that reached states in $A$ are collected in the buffer set $B$ first and \emph{do not}
go directly into the responsible component $(\alpha_i, t_i)$, \emph{unless} it becomes
empty during $\op{step}$ (which indicates a breakpoint).

As all states in $A$ are accepting, they all move to the left child during $\op{step}$,
whereas the right child becomes empty. Consequently, the right empty set is removed,
whereas its rank (which is more important) is preserved by assigning it to the left child
containing all the states. This rank, by definition, is in the set $G_i$, signalling a
good event. On the other hand, if the tracked states have no successors, both child sets
will be empty during $\op{step}$, resulting in $\op{prune}$ removing both of them and the
assigned rank (as well as the fresh rank of the left child set) signalling a bad event,
i.e., being in the set $R_i$. But then all other states in $A$ that are currently
untracked are moved into the component by $\op{relocate}$ with a fresh rank of low
importance, thereby ``restarting'' the track set with the states from the background
set. The $\op{merge}$ operation on a single set clearly is always trivial (i.e. can be
skipped).

In practice, the separation of successors into two sets during $\op{step}$ and most of the
mechanics during $\op{prune}$ can be optimized away by making the current rank of the
single set in the tuple always signal green, unless it becomes empty, when it must
signal red. Thereby we can use only a single rank (assigned by $\alpha_i$) for all states
in $A$. If the track set in the tuple eventually never becomes empty, its importance will
increase by following the construction and eventually its green signal will be the most
important in the transition, unless there exists a different more important rank that also
eventually always signals green, so the DPA is guaranteed to accept if there is an
accepting run in the accepting SCCs of the Büchi automaton. On the other hand, if there is
no such run, then the determinization component that runs the breakpoint construction will
become empty and will be ``restarted'' infinitely often and therefore its current rank
will infinitely often signal red and thus will never witness an accepting run.

 \section{Details on the deterministic SCC heuristic}
\label{app:detsccs}

In this section we describe how the determinization construction can be simplified for
SCCs of the NBA that are already deterministic. For this section, fix a deterministic SCC
$D$ of the Büchi automaton. The idea is based on the fact that no run of the Büchi
automaton that currently is in $D$ can ever split up into multiple runs as
long as its states remain in $D$.

Fix (conceptually) a set of runs currently in $D$, i.e., take the set of states
occupied by these runs. Then track for these states which successors also remain in $D$
(if a run splits up, then all but one successors go to a different SCC and are not
tracked). Clearly, as long as the tracked set is not empty, there are still runs in $D$.
Furthermore, if an accepting state is in this set, then it was reached by one of the runs
we are currently tracking. Therefore, if there are accepting states in this set infinitely
often, then at least one of the finitely many runs tracked in the set must visit accepting
states infinitely often, i.e. there exists an accepting run.

While we are tracking this set, some runs from other SCCs may enter $D$ and reach the same
or also different states. Whenever this happens, we start a new tracking set with a lower
priority, which means that whenever both sets reach the same state $q$, only the set with
higher priority gets state $q$. Thus, we keep only one copy of each state and at
the same time lower priority sets never interfere with higher priority sets. Lower
priority sets therefore only track runs that visit states which are not visited by runs
tracked in sets with higher priority at the same time.

Notice how this concept can be easily represented in our modularized construction (see
Appendix~\ref{app:modular}). Assume that some determinization component $(\alpha_i, t_i)$
is responsible for the deterministic SCC $D$ of the NBA, i.e. we have that $\op{prt}(i) =
D$. We can easily adapt the $\op{step}$ operation to \emph{not} split up the set of
successors, so that together with the $\op{relocate}$ operation we get exactly the
behaviour described above, i.e., older sets are located to the left in $t_i$ and are
favoured by the normalized successor calculation, whereas new sets are introduced on the
right by $\op{relocate}$. These sets contain only states that were not reached by the
other sets in the component and get a fresh rank with low importance. States that leave
$D$ are also automatically taken care of by the construction during $\op{relocate}$.
Whenever a tracked set becomes empty, it will be removed by $\op{prune}$ and by
definition, its rank will be in the set $R_i$ representing a bad event. For a rank to
signal a good event, we modify the construction for this component to reflect the
reasoning above, such that a rank $r$ is put into the set $G_i$ whenever the set $S_{i,j}$
in $t_i$ with $\alpha_i(j) = r$ contains an accepting state. The operations $\op{merge}$
and $\op{normalize}$ work as before.

Observe that by disabling the splitting of sets during $\op{step}$ we can hope to
construct less different partitionings of the states in $D$, as new sets only are
introduced by runs coming from outside of $D$. Furthermore, the assigned ranks in
the tuple $t_i$ are always ascending (because the new sets with low importance are always
added on the right), so that the importance order and tuple order agree. This property
restricts the worst-case upper bound on the state space for the
component $(\alpha_i, t_i)$, as only
ascending rankings $\alpha_i$ are used.

When applying this heuristic for the determinization of limit-deterministic automata, one
obtains a variant of the construction described in  \cite{esparza2017ltl}, in
which the states in the deterministic components are organized in a
tuple sorted by the order in which the corresponding run entered the
deterministic part. In our construction, we use a tuple of sets
instead, tracking those runs that enter the deterministic part at the
same time in one set. While the worst-case upper bound for the number
of resulting macrostates is better for the construction in
\cite{esparza2017ltl}, in our experiments we found that using tuples
of sets also works well in the examples.

 \newpage
\section{Details on smart successor selection}
\label{app:smartsucc}

In this section, we describe in more detail how the freedom of choice for the successor
given by the $\op{merge}$ operation of the determinization construction from
\cite{unidet} can be exploited in a computationally feasible way, i.e., how to quickly
search for a candidate successor macrostate for some transition during state exploration of the
DPA in the set of already existing macrostates.

The crucial idea is to structure the set of already visited states in a reasonable way such
that macrostates which share a common prefix in the trie are related wrt.\ the tree
interpretation of the macrostates. This tree is derived from a ranked slice in the
following way. Let $(\alpha, t)$ be a ranked slice with $t=(S_1, \ldots, S_n)$. The tuple
index of the \emph{parent} of $S_{i}$ is the closest index to the right of $i$ that has a
smaller rank, formally defined as
$\parent(i) := \min_{i < k \leq n}\{k \mid \alpha(k) < \alpha(i)\}$.
The ordered tree induced by $\parent$, with siblings in tuple index order, is called the
\emph{rank-tree} of $(\alpha,t)$ (more on this correspondence of ranked slices and trees
can be found in \cite{unidet}).

For example, in the macrostate $(\{q_3,q_4\}^4,\{q_2\}^2,\{q_5,q_6\}^3,\{q_1\}^1)$ from
Figure~\ref{fig:trie}(b), the rightmost index 4 can have no parent and is the root. Its rank 1
is the closest smaller rank for positions 2 and 3 (that have ranks 2 and 3, respectively),
hence we have $\parent(2) = \parent(3) = 4$, whereas the position 1 (with rank 4) has position 2 (with rank 2) as parent, that is $\parent(1) = 2$. The edges shown in Figure~\ref{fig:trie}(a)--(c) correspond to the edges in this tree interpretation, pointing from a parent to its direct descendants.
In
Figure~\ref{fig:trieword}, the trees on the bottom are actually such macrostates, in this
interpretation as trees. The sequence of sets in the upper part of the figure is explained later.

Rank-trees have the property that parents have smaller ranks than children and sibling
nodes have ascending ranks, which easily follows from the definition. So one can interpret
ranks as the order of creation of nodes in the tree, smaller ranks meaning that the
corresponding tree node was created earlier (in practice, multiple nodes that are created
at the same time are eventually assigned different ranks during the $\op{normalize}$ step
of the determinization construction, fixing an arbitrary creation order for them).

We now explain how such rank trees can be encoded by a single sequence of sets (with ranks implicitly encoded), which is the basis for the data structure that we are using.
We present it here for the case of a single tuple, but it can be adapted to the
modularized variant in a straight-forward way.

First, replace each set in the tuple with a
set containing all states in the subtree of the corresponding rank-tree node. Then, sort
these sets in ascending rank order and drop the ranks. For the macrostate
$(\{q_3,q_4\}^4,\{q_2\}^2,\{q_5,q_6\}^3,\{q_1\}^1)$ from Figure~\ref{fig:trie}(b), the resulting sequence is
$\{q_1,\ldots, q_6\}, \{q_2,q_3,q_4\}, \{q_5,q_6\}, \{q_3,q_4\}$.
The resulting sequence of sets $S_1,S_2,\ldots,S_n$ (which are different from the sets in
the macrostate tuple) is essentially just a word over the alphabet $2^Q$ which satisfies
the property that for all $i>1$, it holds that $\emptyset \neq S_i \subset S_k$ for some
$k<i$ and all sets are pairwise different.

Notice that we can interpret such sequences of sets operationally, as a ``recipe''
describing how one can construct a rank-tree, when first starting just with the set of all
states that are located in its nodes. The first set of the sequence is the set of all
states that exist in the tree, and can be seen as a trivial rank-tree
with all states in the root (which must have the rank 1). The following sets in the
sequence can be read as a list of successive instructions, each set $S$ essentially
saying: ``find the current tree node that contains the states in $S$, create a new child
node of that node (assigning it the new highest rank so far) and move the states in $S$ to
the new node''. Thereby, each additional set in the list ``refines'' the tree, until we
get the exact rank-tree we are describing. Notice that each prefix of this sequence also
describes a valid rank-tree with the same states, but with less nodes over which those
states are distributed (see Figure~\ref{fig:trieword}).

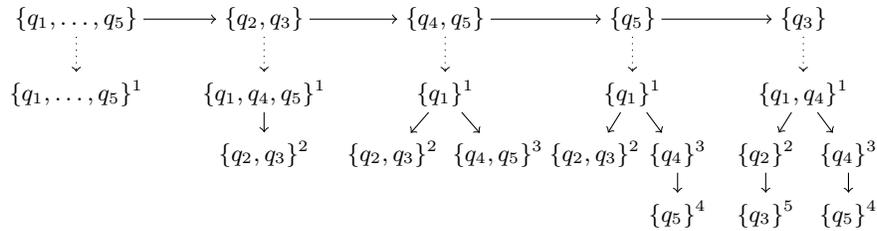
\begin{figure}[htbp]
  \centering
  \begin{tikzpicture}[baseline={([yshift=-.5ex]current bounding box.center)},
    shorten >=1pt,node distance=3mm,inner sep=2pt,auto]

    \node                  (s1) {$\{q_1,\ldots,q_5\}$};
    \node[right=1cm of s1] (s2) {$\{q_2,q_3\}$};
    \node[right=1.2cm of s2] (s3) {$\{q_4,q_5\}$};
    \node[right=1.5cm of s3] (s4) {$\{q_5\}$};
    \node[right=1.5cm of s4] (s5) {$\{q_3\}$};

    \node[below=5mm of s1] (r11) {$\{q_1,\ldots,q_5\}^1$};

    \node[below=5mm of s2] (r21) {$\{q_1,q_4,q_5\}^1$};
    \node[below=of r21] (r22) {$\{q_2,q_3\}^2$};

    \node[below=5mm of s3] (r31) {$\{q_1\}^1$};
    \node[below=of r31,xshift=-7mm] (r32) {$\{q_2,q_3\}^2$};
    \node[below=of r31,xshift=7mm] (r33) {$\{q_4,q_5\}^3$};

    \node[below=5mm of s4] (r41) {$\{q_1\}^1$};
    \node[below=of r41,xshift=-5mm] (r42) {$\{q_2,q_3\}^2$};
    \node[below=of r41,xshift=6mm] (r43) {$\{q_4\}^3$};
    \node[below=of r43] (r44) {$\{q_5\}^4$};

    \node[below=5mm of s5] (r51) {$\{q_1,q_4\}^1$};
    \node[below=of r51,xshift=-5mm] (r52) {$\{q_2\}^2$};
    \node[below=of r51,xshift=6mm] (r53) {$\{q_4\}^3$};
    \node[below=of r52] (r54) {$\{q_3\}^5$};
    \node[below=of r53] (r55) {$\{q_5\}^4$};

    \path[->] (s1) edge (s2) (s2) edge (s3) (s3) edge (s4) (s4) edge (s5);
    \draw (s1) edge[->,dotted] (r11);
    \draw (s2) edge[->,dotted] (r21);
    \draw (s3) edge[->,dotted] (r31);
    \draw (s4) edge[->,dotted] (r41);
    \draw (s5) edge[->,dotted] (r51);

    \path[->] (r21) edge (r22);
    \path[->] (r31) edge (r32) (r31) edge (r33);
    \path[->] (r41) edge (r42) (r41) edge (r43) (r43) edge (r44);
    \path[->] (r51) edge (r52) (r51) edge (r53) (r52) edge (r54) (r53) edge (r55);
  \end{tikzpicture}

  \caption{Example illustrating the semantics of the set sequences which
  describe macrostates in terms of refining their rank-tree. The illustrated rank-trees at
  the bottom correspond to the different prefixes of the set sequence at the top.
  }
  \label{fig:trieword}
\end{figure}

Our approach now is to manage tries that store such sequences of sets, i.e., each trie node
is labelled by a set $S \subseteq Q$ and a marker that indicates whether the sequence from
the root down to the current node is stored in the trie (corresponds to a macrostate of the DPA that is constructed). The words that are added into
this trie are exactly such sequences of sets as described above, so each trie node
corresponds to the macrostate that is described by the set sequence from the root down to
this node. In Figure~\ref{fig:trieword}, the set sequence at the top would be a single
branch of the trie and each set would be a node of the trie that represents a specific
macrostate, which is illustrated below the trie-node as the corresponding rank-tree.

To be more precise, for each set $U \subseteq Q$ of NBA states there is a separate trie
(with the root labelled by $U$) that can be used to store all possible ranked slices (i.e.
macrostates) with exactly the states from $U$, i.e., all macrostates $(\alpha, t)$ with
$Q_t = U$. In practice, these tries are constructed on-the-fly to contain just the
branches corresponding to macrostates that were already added to the DPA (see
Figure~\ref{fig:trie} for an illustration).

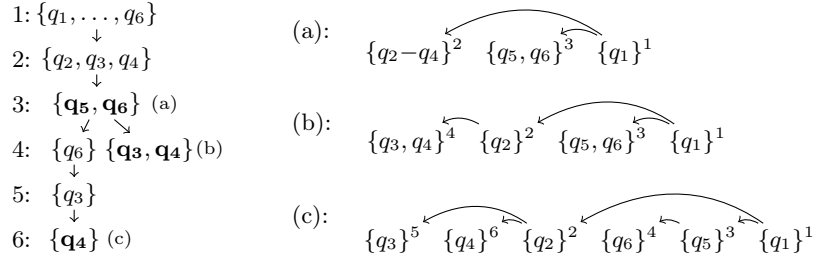
\begin{figure}[htbp]
  \begin{tabular}{p{0.3\textwidth}p{0.7\textwidth}}
    \begin{minipage}{0.3\textwidth}
    \begin{tikzpicture}[baseline={([yshift=-.5ex]current bounding box.center)},
      shorten >=1pt,node distance=6mm,inner sep=1pt,on grid,auto]
      \node              (i1) {1:};
      \node[below=of i1] (i2) {2:};
      \node[below=of i2] (i3) {3:};
      \node[below=of i3] (i4) {4:};
      \node[below=of i4] (i5) {5:};
      \node[below=of i5] (i6) {6:};

      \node[right=of i1,xshift=4mm] (s1) {$\{q_1,\ldots,q_6\}$};
      \node[below=of s1] (s2) {$\{q_2,q_3,q_4\}$};
      \node[below=of s2] (s3) {$\mathbf{\{q_5,q_6\}}$};
      \node[below=of s3,xshift=-3mm] (s4) {$\{q_6\}$};
      \node[below=of s4] (s5) {$\{q_3\}$};
      \node[below=of s5] (s6) {$\mathbf{\{q_4\}}$};
      \node[right=of s4,xshift=4mm] (s7) {$\mathbf{\{q_3,q_4\}}$};

      \node[right=of s3,xshift=3mm] (m1) {\scriptsize (a)};
      \node[right=of s7,xshift=2mm] (m2) {\scriptsize (b)};
      \node[right=of s6,xshift=0mm] (m3) {\scriptsize (c)};

      \path[->]
        (s1) edge (s2)
        (s2) edge (s3)
        (s3) edge (s4)
        (s4) edge (s5)
        (s5) edge (s6)
        (s3) edge (s7)
          ;
    \end{tikzpicture}
    \end{minipage}
      &
    \begin{minipage}{0.7\textwidth}
      (a):
      \quad
    \begin{tikzpicture}[baseline={([yshift=-.5ex]current bounding box.center)},
      shorten >=1pt,node distance=2mm,inner sep=1pt,auto]
      \node              (s1) {$\{q_2{-}q_4\}^2$};
      \node[right=of s1] (s2) {$\{q_5,q_6\}^3$};
      \node[right=of s2] (s3) {$\{q_1\}^1$};

      \path[->]
        (s3) edge[bend right,out=-35] (s1)
        (s3) edge[bend right] (s2)
          ;
    \end{tikzpicture}

      \vspace{3mm}
      (b):
      \quad
    \begin{tikzpicture}[baseline={([yshift=-.5ex]current bounding box.center)},
      shorten >=1pt,node distance=2mm,inner sep=1pt,auto]
      \node              (s1) {$\{q_3,q_4\}^4$};
      \node[right=of s1] (s2) {$\{q_2\}^2$};
      \node[right=of s2] (s3) {$\{q_5,q_6\}^3$};
      \node[right=of s3] (s4) {$\{q_1\}^1$};

      \path[->]
        (s4) edge[bend right,out=-35] (s2)
        (s2) edge[bend right] (s1)
        (s4) edge[bend right] (s3)
          ;
    \end{tikzpicture}

      \vspace{3mm}
      (c):
      \quad
    \begin{tikzpicture}[baseline={([yshift=-.5ex]current bounding box.center)},
      shorten >=1pt,node distance=2mm,inner sep=1pt,auto]
      \node              (s1) {$\{q_3\}^5$};
      \node[right=of s1] (s2) {$\{q_4\}^6$};
      \node[right=of s2] (s3) {$\{q_2\}^2$};
      \node[right=of s3] (s4) {$\{q_6\}^4$};
      \node[right=of s4] (s5) {$\{q_5\}^3$};
      \node[right=of s5] (s6) {$\{q_1\}^1$};

      \path[->]
        (s6) edge[bend right] (s5)
        (s6) edge[bend right,out=-35] (s3)
        (s5) edge[bend right] (s4)
        (s3) edge[bend right] (s2)
        (s3) edge[bend right,out=-35] (s1)
          ;
    \end{tikzpicture}
    \end{minipage}
  \end{tabular}
  \caption{The depicted trie on the left contains three macrostates that correspond to the
            marked trie nodes.
           The macrostate-tuples are depicted on the right, with ranks as
           superscripts and rank-tree edges (defined by parent index function $\parent$)
           connecting parents with children sets.}
  \label{fig:trie}
\end{figure}

Now we explain how to use these tries to aid the successor selection.
Recall that in a transition of the determinization construction, during the $\op{merge}$
operation adjacent sets can be grouped into intervals and merged together. The
nondeterminism in the choice of intervals in this step allows for the multiple valid
successors. In the view of macrostates as rank-trees, it is permitted to merge, for
example, a whole subtree and/or multiple adjacent siblings into one node.

If $k$ is the dominating rank of the transition (obtained after performing
$\op{prune}$), the constraints on $\op{merge}$ ensure that certain nodes with more
important (i.e., smaller) ranks, which are higher up in the rank-tree, more
precisely---all nodes with a rank $<k$, cannot be changed at all. They keep exactly the
same states and retain the same rank as before the $\op{merge}$. This is reflected in the
set sequences described above by the fact that all valid successor macrostates agree on
the prefix of length $k-1$, and only differ in the following sets that describe how the
states are further distributed in the rank-tree. This means, that all valid successors
must be located in the sub-trie that is rooted at the trie node reached after this common
prefix of length $k-1$, which already significantly restricts the search space for a
valid successor state.

Furthermore, each state $q$ that is not in a set of the rank-tree that has a rank $< k$
must eventually appear in some set at trie depth $\geq k$. If it does not, then in the macrostate
described along the current trie branch, $q$ was never ``promoted'' downward in the rank-tree,
violating the fact that nodes that have ranks $<k$ must not be modified by the
$\op{merge}$ operation, and especially must not contain additional states.

On the other hand, observe that one may skip the $\op{merge}$ step completely
(or do a trivial $\op{merge}$ operation that collapses no sets),
which corresponds to the successor agreeing with the Muller-Schupp construction.
Because by performing non-trivial merges, the rank of each state $q$ (i.e., the rank of the
set containing $q$ in the macrostate) may only decrease, the Muller-Schupp successor
carries important information, namely for each state the highest rank of a node (of the
rank-tree) that can have $q$ in its subtree, if the resulting macrostate shall be a valid
successor. Translated to the trie view, this means that for each state of the NBA there is a deepest
level of the trie where it may appear along the corresponding set sequence.

We use the observations above in the following way. During a transition, we first compute
the Muller-Schupp successor (i.e., skip $\op{merge}$) and transform the resulting
macrostate into a sequence of sets, as described above.

First, we check whether the trie node reached by the $k{-}1$-prefix of the set sequence
exists in the trie. If it does not, then there is no viable successor state existing in
the automaton yet, as they all must share this prefix in the trie. Thus, we proceed by
constructing the ``default'' successor, as prescribed by some $\op{merge}$ strategy that
is given as a parameter by the user (e.g. a merge that corresponds to the construction of
Safra, Muller-Schupp, etc.), add this macrostate to the DPA and insert its set sequence
into the trie, i.e., add the branch labelled by those sets and mark the final trie node
reached by this set sequence, indicating that the corresponding macrostate is in the DPA.

If the trie node corresponding to the prefix does already exist in the trie, we search the
trie using a depth-first search. This DFS follows existing branches in the subtrie, as
long as they don't violate the constraints above (which can be computed from the
Muller-Schupp successor we already constructed). Whenever a trie node is reached where
these constraints are violated, the exploration of its sub-trie can be aborted, as the
macrostates represented further below may not arise by a valid $\op{merge}$ operation in
the current transition.

Each trie node which does not violate the constraints, and is also marked, corresponds to
a macrostate which already exists in the automaton and is a candidate for a valid
successor. What remains to be ensured is that the candidate macrostate really can be
obtained by merging adjacent non-empty sets. This is done using a simple unidirectional
scan which compares the Muller-Schupp successor (which we constructed in the beginning) to
the current candidate macrostate (which already existed and just needs to be looked up).
If a valid successor is identified during this exploration, then the DPA just gets a new
transition to this already existing macrostate. Otherwise, we construct the ``default''
successor, add it to the DPA as the target of the current transition and also insert it
into the trie, as described above.

From the practical point of view, all required steps can be carried out efficiently using
bit-set operations, i.e. all set operations can be seen as constant-time operations. For
all checks and computations only a linear (in the number of sets in a macrostate) number
of set operations is required. Hence, the overall slowdown incurred by this optimized
successor selection is very moderate.
 \newpage
\section{Additional examples}
\label{app:experiments}

The example in Figure~\ref{fig:badforheu} illustrates a case where heuristics presented in
Section~\ref{subsec:NBASCCs}, which exploit the SCC-structure of the NBA, can be detrimental
to state reduction and therefore using them requires greater care. One of the reasons is that by adding
more determinization components in the modularized construction (Appendix~\ref{app:modular}),
it is possible that additional new macrostates are constructed due to the relocation of
reached states of the NBA between the different tuples. Therefore, SCC-based heuristics
should be only applied with sufficiently large SCCs, where a large reduction can be
expected due to the separated handling.

\begin{figure}[h]
    \begin{center}
      \begin{tabular}{lcr}
        $\mathcal{B}(n):$ &
    \begin{tikzpicture}[baseline={([yshift=-.5ex]current bounding box.center)},
      shorten >=1pt,node distance=1cm,inner sep=1pt,on grid,auto]
      \node[state,initial left,initial text=] (q1) {$q_1$};
      \node[state] (q2) [right=of q1] {$q_2$};
      \node (qi) [right=of q2] {\ldots};
      \node[state] (qn) [right=of qi] {$q_n$};
      \node[state,accepting] (qf) [right=of qn] {$q_{\scriptscriptstyle F}$};

      \path[->]
        (q1) edge node {$\scriptstyle 2$} (q2)
        (q2) edge node {$\scriptstyle 3$} (qi)
        (qi) edge node {$\scriptstyle n$} (qn)
        (q1) edge[bend left,out=90,in=90] node {$\scriptstyle 1$} (qf)
        (q1) edge[bend left,out=50,in=130] node {$\scriptstyle 3$} (qi)
        (q1) edge[bend left,out=70,in=110] node {$\scriptstyle n$} (qn)
        (q2) edge[bend left,out=50,in=130] node {$\scriptstyle n$} (qn)
        (q2) edge[bend left,out=70,in=110] node {$\scriptstyle 2$} (qf)
        (qn) edge[bend left,out=50,in=130] node {$\scriptstyle n$} (qf)
        (q1) edge[loop below] node {$\scriptstyle\Sigma\setminus\{1\}$}(q1)
        (q2) edge[loop below] node {$\scriptstyle\Sigma\setminus\{2\}$}(q2)
        (qn) edge[loop below] node {$\scriptstyle\Sigma\setminus\{n\}$}(qn)
        (qf) edge[loop below] node {$\scriptstyle\Sigma\setminus\{\#\}$}(qf)
        ;
    \end{tikzpicture} &
  \quad\quad
  \def\arraystretch{1.2}
  \setlength{\tabcolsep}{3pt}
    \begin{tabular}{r|l|ll}
      $n$ & \texttt{spot} & def. & +sep \\
      \hline
      5 & 62  & 62
              & 454 \\
      6 & 126 & 126
              & 2607 \\
      7 & 254 & 254
              & 17612 \\
    \end{tabular}
      \end{tabular}
      \vspace{-5mm}
    \end{center}

  \caption{
  $\mathcal{B}(n)$ is a family of automata with $n+1$ states for each $n$,
  where enabling some of the heuristics that enforce separated handling of SCCs of the NBA leads to more states.
  The table shows the resulting automata sizes from our
  prototype with just basic optimizations, and the same configuration with
  separation of all SCCs into different determinization components (+sep) enabled. For any
  choice of the $\op{merge}$ step, useless permutations of the states are introduced by
  the enforced separation, leading to the blowup.
  }

  \label{fig:badforheu}
\end{figure}
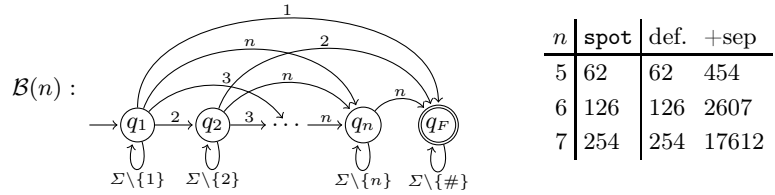

The example in Figure~\ref{fig:goodtopo} demonstrates how well the topological
optimization complements the Muller-Schupp variant of the $\op{merge}$ step in the
determinization construction from Section~\ref{sec:construction}.

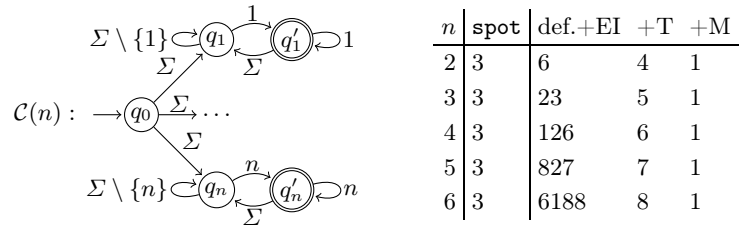
\begin{figure}[h]
  \centering
    \begin{tabular}{lcr}
        $\mathcal{C}(n):$ &
    \begin{tikzpicture}[baseline={([yshift=-.5ex]current bounding box.center)},
      shorten >=1pt,node distance=1cm,inner sep=1pt,on grid,auto]
      \node[state,initial left,initial text=] (q0) {$q_0$};
      \node[right=of q0] (q2) {$\ldots$};
      \node[state,above=of q2] (q1) {$q_1$};
      \node[state,accepting,right=of q1] (q1a) {$q_1'$};
      \node[state,below=of q2] (qn) {$q_n$};
      \node[state,accepting,right=of qn] (qna) {$q_n'$};

      \path[->]
          (q0) edge node {$\Sigma$} (q1)
          (q0) edge node {$\Sigma$} (q2)
          (q0) edge node {$\Sigma$} (qn)
          (q1) edge[loop left] node {$\Sigma\setminus\{1\}$} (q1)
          (qn) edge[loop left] node {$\Sigma\setminus\{n\}$} (qn)
          (q1) edge[bend left] node {$1$} (q1a)
          (qn) edge[bend left] node {$n$} (qna)
          (q1a) edge[bend left] node {$\Sigma$} (q1)
          (qna) edge[bend left] node {$\Sigma$} (qn)
          (q1a) edge[loop right] node {$1$} (q1a)
          (qna) edge[loop right] node {$n$} (qna)
          ;
    \end{tikzpicture} & \quad\quad
  \def\arraystretch{1.2}
  \setlength{\tabcolsep}{3pt}
    \begin{tabular}{r|l|lll}
      $n$ & \texttt{spot} & def.+EI & +T & +M \\
      \hline
      2 & 3 & 6 & 4 & 1 \\
      3 & 3 & 23 & 5 & 1 \\
      4 & 3 & 126 & 6 & 1 \\
      5 & 3 & 827 & 7 & 1 \\
      6 & 3 & 6188 & 8 & 1 \\
    \end{tabular}
    \end{tabular}
  \caption{
    $\mathcal{C}(n)$ is a family of universal automata with $2n+1$ states for each $n$ on which
    using Muller-Schupp successors generates many useless states, which can be
    removed again using the topological optimization, as shown in the table. Together with
    our minimization we obtain the optimal automata.
  }
  \label{fig:goodtopo}
\end{figure}
 }

\end{document}